\documentclass[pra,longbibliography,twocolumn,showpacs,nofootinbib,superscriptaddress,notitlepage]{revtex4-1}
\usepackage{amsmath}
\usepackage{amssymb,bm}
\usepackage{amsthm}

\newtheorem{innercustomgeneric}{\customgenericname}
\providecommand{\customgenericname}{}
\newcommand{\newcustomtheorem}[2]{%
  \newenvironment{#1}[1]
  {%
   \renewcommand\customgenericname{#2}%
   \renewcommand\theinnercustomgeneric{##1}%
   \innercustomgeneric
  }
  {\endinnercustomgeneric}
}

\newcustomtheorem{customthm}{Theorem}
\newcustomtheorem{definitionBob}{Definition}

\usepackage{color,dsfont} 
\usepackage{graphicx, caption} 
\usepackage{ragged2e}
\usepackage[colorlinks=true, hyperindex, breaklinks, linkcolor=blue, urlcolor=blue, citecolor=blue]{hyperref} 
\usepackage[normalem]{ulem}
\usepackage[capitalise]{cleveref}
\usepackage{mathrsfs}
\usepackage{subcaption}
\usepackage{mathtools}
\usepackage{float}
\usepackage{verbatim}
\usepackage{latexsym}
\usepackage{amsmath}
\usepackage{amssymb}
\usepackage{setspace}
\usepackage{amsfonts}
\usepackage{stmaryrd}
\usepackage{xcolor}
\usepackage{enumitem}
\usepackage{hhline}

\newtheorem{theorem}{Theorem} \newtheorem{lemma}{Lemma}
 
\newtheorem{claim}{Claim} \newtheorem{definition}{Definition}

\newcommand{\ket}[1]{|#1\rangle}  
 
\newcommand{\codepar}[1]{\ensuremath{[\![#1]\!]}}

\crefname{equation}{Eq.\!}{Eqs.\!}
\crefname{figure}{Fig.\!}{Figs.\!}
\begin{document}

\title{Flag fault-tolerant error correction, measurement, and quantum computation \\ for cyclic CSS codes}

\author{Theerapat Tansuwannont}
\email{ttansuwannont@uwaterloo.ca}
\affiliation{
    Institute for Quantum Computing and Department of Physics and Astronomy,
    University of Waterloo,
    Waterloo, Ontario, N2L 3G1, Canada
    }
\author{Christopher Chamberland}
\email{christopher.chamberland@ibm.com}
\affiliation{
   IBM T. J. Watson Research Center, Yorktown Heights, NY, 10598, United States
    }
\affiliation{
    Institute for Quantum Computing and Department of Physics and Astronomy,
    University of Waterloo,
    Waterloo, Ontario, N2L 3G1, Canada
    }
    
    \author{Debbie Leung}
\email{wcleung@uwaterloo.ca}
\affiliation{
   Institute for Quantum Computing and Department of Combinatorics and Optimization,
    University of Waterloo,
    Waterloo, Ontario, N2L 3G1, Canada
    }

\begin{abstract}
Flag qubits have recently been proposed in syndrome extraction circuits to detect high-weight errors arising from fewer faults. The use of flag qubits allows the construction of fault-tolerant protocols with the fewest number of ancillas known to-date. In this work, we prove some critical properties of CSS codes constructed from classical cyclic codes that enable the construction of a flag fault-tolerant error correction scheme. We then develop fault-tolerant protocols as well as a family of circuits for flag fault-tolerant error correction and operator measurement, requiring only four ancilla qubits and applicable to cyclic CSS codes of distance 3. The measurement protocol can be further used for logical Clifford gate implementation via quantum gate teleportation. We also provide examples of cyclic CSS codes with large encoding rates.
\end{abstract}


\maketitle

\section{Introduction}
\label{sec:Intro}%
Fault-tolerant quantum computation is an essential component in building a large scale quantum computer.  It enables \emph{arbitrarily} low logical error rates, despite all operations (including those used to perform error correction) may be noisy, as long as the noise strength is below a \emph{constant but sufficiently small} threshold value \cite{Shor96,AB97,Preskill98,KLZ96}. The value of the threshold depends on several factors, including the underlying quantum error correcting code, the design of the fault-tolerant gadgets and the error correction protocol, the speed of quantum measurements and classical processing of the error syndromes, and the underlying physical noise. Currently, the surface code appears to be a strong candidate for fault-tolerant quantum computation given its high threshold value as well as the geometric locality of the gates used in the syndrome extraction circuits \cite{DKLP02,FMMC12,TS14,FowlerAutotune}.

Meanwhile, low logical error rates requires large qubit and gate overheads \cite{PR12,CJL16b,YoderOverhead}. Therefore, a fault-tolerant protocol that uses fewer ancilla qubits (and thus lower overheads) is easier to realize experimentally. A fault-tolerant protocol limits the number of physical errors in each code block arising from a single fault. Recently, Chao and Reichardt \cite{CR17v1,CR17v2} showed that fault-tolerant error correction (FTEC) as well as fault-tolerant quantum computation can be achieved using only two extra ancilla qubits for perfect distance-3 codes. The idea is to use \emph{flag qubits} to detect high-weight errors arising from a single fault. Furthermore, Reichardt showed that stabilizer measurements with flag qubits for the Steane code can be parallelized to reduce the circuit depths \cite{ReichardtF;agColorCode}. In \cite{CB17}, FTEC protocol using very few flag qubits were developed for several families of stabilizer codes of arbitrary distance. For example, color codes with a hexagonal lattice and arbitrary distance require only four ancilla qubits in the FTEC scheme. The protocol in \cite{CB17} can be used with LDPC (low density parity check) codes to achieve constant overhead \cite{LDPC13,Gottesman13LDPC,TZLDPC14}. Flag qubits were further used to fault-tolerantly prepare magic states with very low overhead compared to previous distillation schemes when Clifford gates are noisy \cite{CC18}. Lastly in \cite{FlagGKPStatePrep}, it was shown how flag qubits can be used to fault-tolerantly prepare GKP states.

The idea behind flag-FTEC \cite{CR17v1} is that high-weight errors arising from a single fault have special structure. Despite their high weight, these errors can be alerted using few flag qubits and distinguished by subsequent syndrome measurements. However, there is no general theory what codes admit the flag technique. An interesting family of quantum codes consists of CSS codes constructed from classical cyclic codes. These codes have cyclic structures, each stabilizer generator is either $X$-type or $Z$-type, and some of these codes have high encoding rates. These properties make them a good choice for fault-tolerant quantum computation (see \cref{sec:ExCyclicCSS}).

In this work, we generalize the flag technique to the family of cyclic CSS codes by exploiting the cyclic structure in the high-weight errors arising from a single fault. We build on the previous flag-FTEC schemes and obtain a new flag-FTEC scheme applicable to cyclic CSS codes of distance 3. In particular, we construct circuits for measuring the error syndromes using flag qubits which require only four ancilla qubits (see \cref{fig:GeneralWflagFig}). The circuit uses a particular ordering of the CNOT gates which is independent of the underlying stabilizer code. Our work further expands the code families where flag-FTEC schemes can be used with very few ancilla qubits. Moreover, the number of required ancilla qubits is independent of the weights of the stabilizers being measured. Finally, we provide a flag fault-tolerant (flag-FT) operator measurement protocol for cyclic CSS codes, which can be further used for Clifford gate implementation and other applications.

The paper is organized as follows: In \cref{sec:ReviewFlag} we review the basic properties of flag error correction and CSS codes. Key definitions which are used in several parts of the paper are introduced. We define the notion of distinguishable errors and consecutive error sets which are key components of our flag-FTEC scheme. We conclude the section by stating \cref{Lemma:Lem1}, an important building block for constructing our flag-FTEC scheme. In \cref{sec:CyclicCodesDistinguishability} we review basic properties of classical cyclic codes, then state \cref{Lemma:Cyclic_gen,Lemma:Lem3}. Using the lemmas, we state and prove \cref{theo:MainTheorem} which is the final ingredient required to construct our flag-FTEC scheme for CSS codes constructed from classical cyclic codes. In \cref{sec:Protocol} we describe the syndrome extraction circuit used in our flag-FTEC protocol and proceed by describing the protocol in detail as well as explaining how it satisfies the fault-tolerance criteria. In \cref{sec:MeasProtocol} we provide a flag-FT measurement protocol for Pauli operators, and its possible applications are discussed in \cref{sec:Applications}. Examples of distance-3 cyclic CSS codes are given in \cref{sec:ExCyclicCSS}. Lastly, we discuss our results and directions for future work in \cref{sec:Discussion}.

\section{Flag error correction with CSS codes}
\label{sec:ReviewFlag}%

CSS codes form one of the most studied families of quantum codes since they have nice properties for fault-tolerant quantum computation. It has been shown recently that the technique of flag-FTEC can be applied to several families of codes \cite{CR17v1,CB17}, but it remains open whether the techniques can also be applied to general CSS codes. In this section, we will analyze the idea behind flag techniques and provide the conditions which make CSS codes suitable for flag-FTEC in \cref{Lemma:Lem1}. This lemma will be a main ingredient for our theorem for cyclic CSS codes in the next section.

We first define CSS codes. They are constructed from classical binary linear codes \cite{MS77} as follows: An $[n,k,d]$ classical linear code $C$ encodes $k$ bits in $n$ and has distance $d$ (the minimum Hamming weight of the codewords). It corrects up to $t = \lfloor (d-1)/2 \rfloor$ errors. The code is defined by the parity check matrix $H$ which consists of $n-k$ independent rows that are orthogonal to every codeword. The dual code $C^\perp$ of $C$ consists of codewords that are orthogonal to all codewords in $C$. Note that $C^\perp$ is generated by $H$, that is, each codeword in $C^\perp$ is a linear combination of rows of $H$.  

A quantum \codepar{n,k,d} stabilizer code \cite{Gottesman96,Gottesman97} encodes $k$ logical qubits in $n$ physical qubits. It is the simultaneous $+1$ eigenspace of $n-k$ commuting, independent, Pauli operators. These Pauli operators multiplicatively generate a group called the \emph{stabilizer group} for the code, and the Pauli operators are called the \emph{stabilizer generators}. The code has distance $d$ (see \cite{Gottesman96}) and it can correct errors acting on up to $t = \lfloor (d-1)/2 \rfloor$ qubits. Let $I,X,Y,Z$ denote the single-qubit Pauli operators. A Pauli operator $P$ on $n$ qubits, given by $P=\bigotimes_{i=1}^n X^{x_i} Z^{z_i}$ up to a phase, has a \emph{symplectic representation} $\sigma(P)$ which is the $2n$-bit string $\sigma(P)=(x_1, \cdots, x_n| z_1, \cdots, z_n)$. The symplectic representation of a stabilizer code is an $(n-k) \times 2n$ binary matrix where the $i^\text{th}$ row is the symplectic representation of the $i^\text{th}$ generator. The CSS codes first proposed in \cite{CS96,Steane96b} can be defined in the stabilizer formalism as follows:

\begin{definition}{\underline{\smash{CSS code}} \cite{CS96,Steane96b,Gottesman97}}
	
	An \codepar{n,k,d} stabilizer code is a \emph{CSS code} if the generators can be chosen such that the code has symplectic representation
	\begin{align}
	\begin{pmatrix}
	A &|& 0 \\
	0 &|& B 
	\end{pmatrix},
	\end{align}
	\label{Def:CSS}%
	where $A$ is an $r_x\times n$ matrix and B is an $r_z\times n$ matrix for some $r_x$ and $r_z$ with $r_x + r_z = n-k$. $A$ and $B$ are called $X$ and $Z$ stabilizer matrices.
\end{definition}
In other words, a CSS code is a stabilizer code whose generators can be chosen to be either tensor products of $I$ and $X$ or of $I$ and $Z$. The generators of $X$-type and $Z$-type are called $X$ and $Z$ stabilizers, respectively. With this choice of generators, the $X$ errors and $Z$ errors can be detected separately.

\begin{theorem}{\underline{\smash{CSS code construction}} \cite{Gottesman97}}
	
	Let $C_x$ be an $[n,k_x,d_x]$ classical linear code with parity check matrix $H_x$ and $C_z$ be an $[n,k_z,d_z]$ classical linear code with parity check matrix $H_z$. Suppose that $H_x^T H_z = 0$, or equivalently, $C_x^\perp \subseteq C_z$. Then, the following binary matrix
	\begin{align}
	\begin{pmatrix}
	H_x &|& 0  \\
	0 &|& H_z
	\end{pmatrix},
	\end{align}
	\label{Theorem:CSS}%
	is the symplectic representation of an \codepar{n,k,d} stabilizer code $C$ with $k=k_x + k_z - n$ and $d\geq\min\{d_x,d_z\}$. 
\end{theorem}

In an EC protocol, the syndrome measurement corresponds to the measurement of all stabilizer generators. Consider an \codepar{n,k,d} CSS code which can correct errors of maximum weight $t =\lfloor (d-1)/2 \rfloor$. Each generator is either $X$-type or $Z$-type stabilizer, and it acts non-trivially on $m$ qubits where $m \in \{1,\cdots,n\}$. We can assume that, up to qubit permutations, the stabilizer being measured is of the form $I^{\otimes n-m}\otimes X^{\otimes m}$ or $I^{\otimes n-m}\otimes Z^{\otimes m}$. The ideal circuits for measuring weight-$m$ $X$ stabilizers and weight-$m$ $Z$ stabilizers are shown in \cref{fig:StabNonFT}. 

\begin{figure}[htbp]
	\centering
	\begin{subfigure}{0.4\textwidth}
		\includegraphics[width=\textwidth,trim = {0 0.5cm 0 0}, clip]{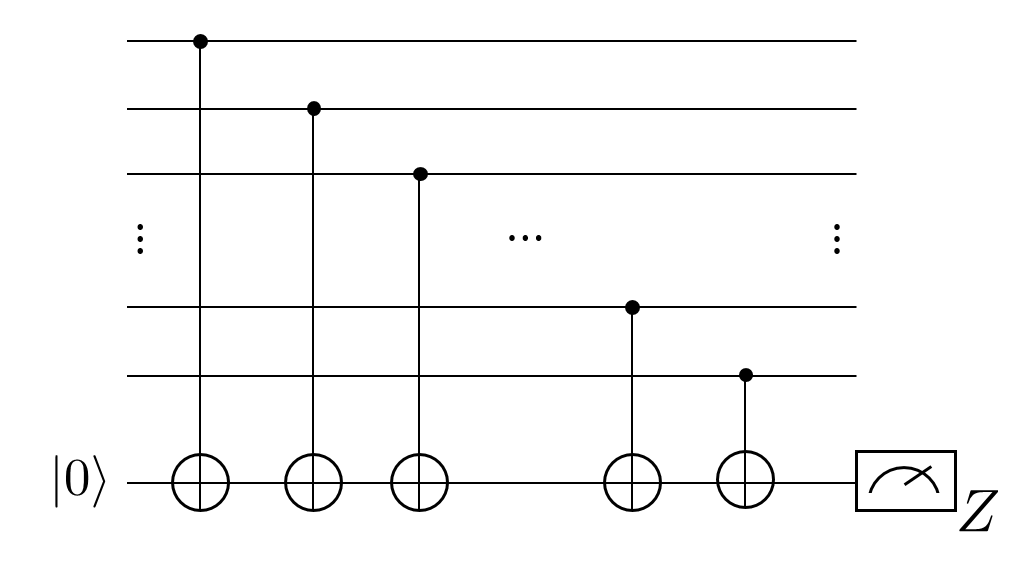}
		\caption{}
		\label{fig:StabNonFT}
	\end{subfigure}	
	\begin{subfigure}{0.25\textwidth}
		\includegraphics[width=\textwidth,trim = {0 0.8cm 0 0}, clip]{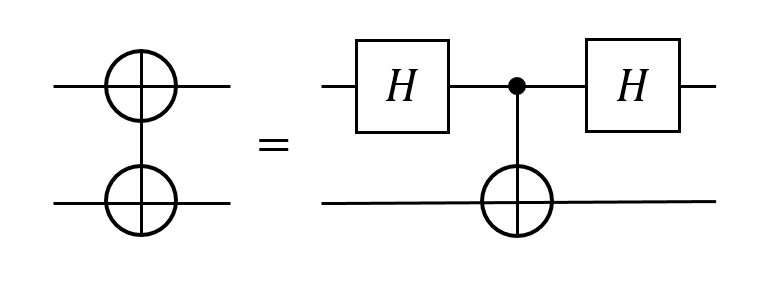}
		\caption{}
		\label{fig:XNOTgates}
	\end{subfigure}
	\caption{In \cref{fig:StabNonFT}, we illustrate the ideal circuit for measuring a weight-$m$ $Z$ stabilizer. Only the qubits with non-trivial support on the stabilizer being measured are shown. The measurement is performed on the eigenbasis of $Z$ operator (i.e., the computational basis), and the measurement results 0 and 1 correspond to the +1 and -1 eigenvalues of $Z$. The circuit for measuring the $X$ stabilizers is obtain by replacing the CNOT gates with the gates shown in \cref{fig:XNOTgates}.}
	\label{fig:CircuitStabMeas}
\end{figure}

However, the EC protocol involving the aforementioned circuit has a drawback. Suppose the circuit is not perfect, and each location (a state preparation step, a gate, or a measurement) can have a fault. Suppose that $v \leq t$ faults happen. In some cases, these $v$ faults can result in an error of weight greater than $t$ in the output state of the circuit, which may not be correctable anymore. This circuit spreads errors and is not generally suitable for building EC protocols with an important property called \emph{fault tolerance} \cite{AGP06}, defined as follows:

\begin{definition}{\underline{\smash{Fault-tolerant error correction}} \cite{AGP06}}
	
	For $t = \lfloor (d-1)/2\rfloor$, an error correction protocol using a distance-$d$ stabilizer code $C$ is \emph{$t$-fault-tolerant} if the following two conditions are satisfied:
	\begin{enumerate}
		\item For an input codeword with error of weight $v_{1}$, if $v_{2}$ faults occur during the protocol with $v_{1} + v_{2} \le t$, ideally decoding the output state gives the same codeword as ideally decoding the input state.
		\item For $v$ faults during the protocol with $v \le t$, no matter how many errors are present in the input state, the output state differs from a codeword by an error of at most weight $v$.
	\end{enumerate}
	\label{Def:FaultTolerantDef}
\end{definition}

An error on the input state might have weight $> t$, which means that it is incorrectable. Anyhow, if the number of faults is $v\leq t$, the second condition in \cref{Def:FaultTolerantDef} requires that the state after correction must differ from \emph{any valid codeword} by an error of weight $\leq v$. (One possible way to construct a FTEC protocol satisfying both conditions in \cref{Def:FaultTolerantDef} is using the minimal weight correction, defined later in \cref{def:MinErr}.)

Ideally decoding is equivalent to performing fault-free error correction. The conditions above are simultaneously required in order to ensure that low-weight errors do not spread and become incorrectable as well as to prevent errors from accumulating between different error correction rounds. 

Generally, FTEC protocols may require a lot of ancilla qubits to avoid the spread of errors within a code block. Chao and Reichardt introduced the idea of flag qubits in \cite{CR17v1} to reduce the number of ancilla qubits being used in FTEC. They also provided some circuit constructions to fault-tolerantly extract syndromes for various distance-3 perfect stabilizer codes using only two ancilla qubits. To see how the flag-FTEC works, let us examine the circuit shown in \cref{fig:StabFTwithAncilla} which is modified from the circuit in \cref{fig:StabNonFT}.
\begin{figure}[htbp]
	\centering
	\begin{subfigure}{0.4\textwidth}
		\includegraphics[width=\textwidth,trim = {0 0.5cm 0 0}, clip]{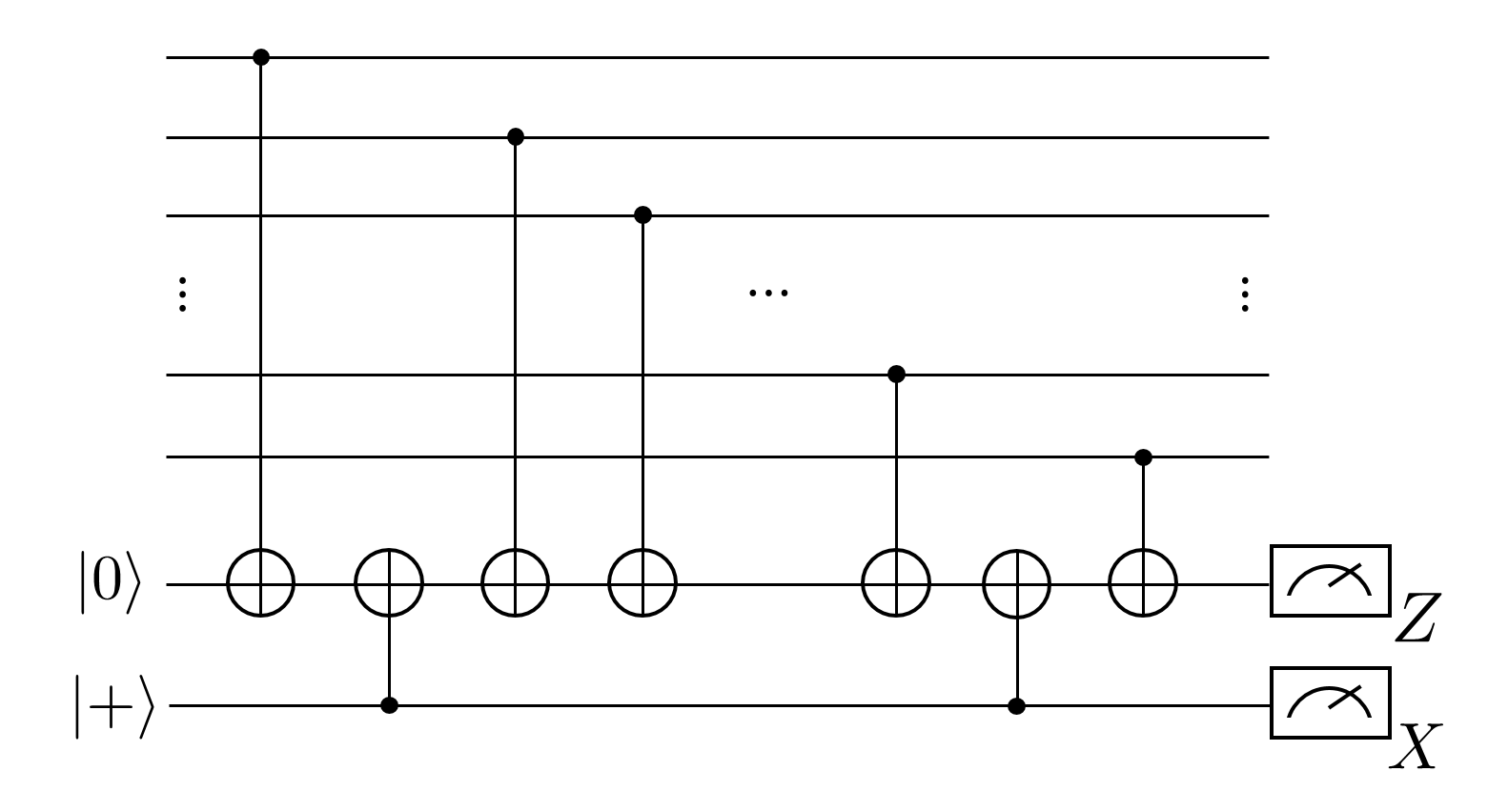}
	\end{subfigure}
	\caption{A circuit obtained from \cref{fig:StabNonFT} by including a flag qubit prepared in the $\ket{+}$ state. The measurement of flag qubit is performed on the eigenbasis of $X$ operator. If a single fault produces an error of weight $>1$ on the data qubit, the outcome of the flag-qubit measurement will be $-1$, otherwise it will be $+1$.}
	\label{fig:StabFTwithAncilla}
\end{figure}

A flag qubit is introduced in \cref{fig:StabFTwithAncilla} to detect a fault that can lead to data error of weight $>1$. If any pair of higher-weight errors detected by the flag qubit are either equivalent (up to multiplication of some stabilizer) or have different syndromes, it is possible to construct a flag-FTEC protocol which corrects higher weight errors (that arise from a single fault) using information from the flag qubit and subsequent syndrome measurements.

The idea of flag-FTEC is further developed in \cite{CB17}, and the general conditions for flag-FTEC applicable to stabilizer codes of arbitrary distance are provided. In particular, the flag-FTEC condition for a stabilizer code which can correct up to one error is as follows:
\begin{definition}{\underline{\smash{Flag-1 FTEC condition}} \cite{CB17}}
	
	Consider a stabilizer code generated by $\{g_1,\dots,g_{n-k}\}$ which can correct up to one error. Let $\mathcal{E}(g_i)$ be the set of all possible errors arising from any single fault that can cause a circuit for measuring $g_i$ to flag. For every generator $g_i$, all pairs of errors in $\mathcal{E}(g_i)$ must either have different syndromes or be equivalent up to multiplication of some stabilizer.
	\label{Def:flagFTECcond}%
\end{definition}
Showing that a code along with appropriate syndrome extraction circuits satisfy the general conditions for flag-FTEC can be quite challenging. Reference \cite{CB17} provides a sufficient condition which implies the general flag-FTEC conditions, and a FTEC protocol using flag qubits and applicable to stabilizer codes of arbitrary distance satisfying such condition was developed. This sufficient condition can be much easier to verify, and several code families were shown to satisfy the general flag-FTEC conditions. However, not all CSS codes satisfy this sufficient condition.

As was shown in \cite{CR17v1}, for codes which do not satisfy the sufficient condition in \cite{CB17}, errors are spread in the measurement circuits in a way that depends on which stabilizer generators are measured, and also on the ordering of the CNOT gates used in the measurement circuits for these generators. Therefore, these specific designs in the protocol may affect the fulfillment of the flag-FTEC conditions. With an appropriate permutation of the CNOT gates of the syndrome extraction circuits, Chao and Reichardt proved that the family of \codepar{2^r-1,2^r-1-2r,3} quantum Hamming codes satisfied the flag-1 FTEC condition. In this work, we prove some properties of cyclic CSS codes and show that it is possible to construct syndrome extraction circuits which satisfy the flag-1 FTEC condition in \cref{Def:flagFTECcond} for cyclic CSS codes of distance 3.

Reference \cite{CB17} develops the notation of $t$-flag circuits and shows that any flag-FTEC protocol will require the use of them. We generalize their definition as follows:

\begin{definition}{\underline{\smash{t-flag circuit}}}
	
	Let $C$ be an \codepar{n,k,d} stabilizer code with generators $g_1,g_2\dots,g_{n-k}$, $P$ be a weight-$m$ Pauli operator which commutes with all $g_i$'s, and ${\cal C}(P)$ be a circuit that implements a projective measurement of $P$ in the absence of faults. We say that ${\cal C}(P)$ is a \emph{$t$-flag circuit} if all of the following holds:
	\begin{enumerate}
		\item the circuit does not flag without faults, and 
		\item the circuit flags whenever a set of $v \leq t$ faults in ${\cal C}(P)$ leads to an error $E$ on the output with $\text{min}_Q(\text{wt}(E Q)) > v$ where the minimization is over $Q \in \langle P, g_1, \cdots, g_{n-k} \rangle$, the group generated multiplicatively by $P$ and the stabilizer generators $g_i$'s.
	\end{enumerate}
	\label{Def:tFlaggedCircuitDef}%
\end{definition}

In this paper, we will use certain properties of cyclic CSS codes to develop a flag-FTEC protocol.  In particular, a \emph{single} fault in the syndrome extraction circuits of CSS codes produces errors with special properties which allow us to distinguish \textit{consecutive} errors. To proceed with the analysis, we introduce some useful definitions and lemmas. We start by the definition of distinguishable errors as follows:

\begin{definition}{\underline{\smash{Distinguishable errors}}}
	
	Let $C$ be an \codepar{n,k,d} stabilizer code and let $E_1$ and $E_2$ be $n$-qubit Pauli errors with syndromes $s(E_1)$ and $s(E_2)$. We say that $E_1$ and $E_2$ are \emph{distinguishable} by $C$ if $s(E_1) \neq s(E_2)$. Otherwise we say that they are \emph{indistinguishable}. In addition, if any pair of errors from an error set $\mathcal{E}$ are distinguishable by $C$, we says that $\mathcal{E}$ is distinguishable by $C$.
	\label{Def:DisErr}%
\end{definition}

The circuit in \cref{fig:StabFTwithAncilla} is a 1-flag circuit since it will flag (the flag-qubit measurement outcome is -1) if there is a single fault causing data error of weight $>1$. From the flag-FTEC condition in \cref{Def:flagFTECcond}, our goal is to distinguish all possible higher-weight errors by subsequent stabilizer measurements. Note that the set of higher-weight errors depends on the choice of generators and the permutation of CNOT gates, and only some choices and permutations will lead to a distinguishable error set. Some CSS codes that satisfy the sufficient condition in \cite{CB17} can be used in a flag-FTEC protocol\footnote{Note that for such codes, the order of the CNOT gates in a $t$-flag circuit is not important.}. However, whether flag-FTEC techniques can be applied to general CSS codes is still unknown. 

Observe that permuting the CNOT gates in the measurement circuit is equivalent to permuting columns of the stabilizer matrices. In order to find CSS code families such that flag-FTEC techniques can be used, we will consider fixing the CNOT gates of syndrome extraction circuits in the \textit{normal permutation}, (i.e., applying CNOT gates from top to bottom as in \cref{fig:StabFTwithAncilla})\footnote{Note that for some specific codes, it is certainly possible to find circuits with fewer ancilla qubits by choosing an appropriate permutation of the CNOT gates.}. Subsequently, we will find conditions that need to be satisfied by the $X$ and $Z$ stabilizer matrices.

Assume that a faulty CNOT gate can cause a 2-qubit error of the form $P_1\otimes P_2$ where $P_1,P_2 \in \{ I,X,Y,Z \}$ are Pauli errors on the control and the target qubits, respectively. Consider a circuit for measuring stabilizers of the form $I^{\otimes n-m}\otimes Z^{\otimes m}$ with the normal permutation of CNOT gates as in \cref{fig:StabFTwithAncilla} where $m\in\{1,\cdots,n\}$. A single fault at a CNOT location can result in the following types of errors:
\begin{enumerate}[label={(\alph*)}]
	\item If an error from a faulty CNOT gate is of the form $P_1 \otimes P_2$ where $P_1\in \{I,X,Y,Z \}$ and $P_2 \in \{ I,X \}$, then the data error is of weight $\leq 1$ and the flag outcome is $+1$. \label{form_a}%
	\item If an error from a faulty CNOT gate is $P_1 \otimes P_2$ where $P_1=I$ and $P_2 \in \{ Y,Z \}$, the data error is of the form $I^{\otimes n-m+c}\otimes Z^{\otimes m-c}$ where $c\in\{1,\cdots,m\}$. In the cases where the data error has weight $>1$, the flag outcome is $-1$.
	\label{form_b}%
	\item If an error from a faulty CNOT is $P_1 \otimes P_2$ where $P_1 \in \{ X,Y,Z \}$ and $P_2 \in \{ Y,Z \}$, the data error is of the form $I^{\otimes n-m+c-1}\otimes P_1 \otimes Z^{\otimes m-c}$ where $c\in\{1,\cdots,m\}$. In the cases where the data error has weight $>1$, the flag outcome is $-1$.
	\label{form_c}%
\end{enumerate} 
Data errors of the form \ref{form_b} or \ref{form_c} arise due to the propagation of $Z$ errors from the target to control qubit of CNOT gates. In addition, if a faulty CNOT gate causes the error $Z\otimes Z$, this can be viewed as an error $I\otimes Z$ caused by the preceding CNOT gate. Let $\mathcal{E}_+$ and $\mathcal{E}_-$ be sets of errors corresponding to the flag outcome $+1$ and $-1$, respectively. Consider an \codepar{n,k,d} CSS code $C$ constructed from two classical codes $C_x$ and $C_z$ as in \cref{Theorem:CSS}. It is obvious that $\mathcal{E}_+$ is distinguishable by $C$ if $d\geq 3$. The distinguishability of errors of the form \ref{form_b} in $\mathcal{E}_-$ depend on the classical code $C_x$. Also, any error of the form \ref{form_c} in $\mathcal{E}_-$ can be considered as a product of an error of the form \ref{form_b} and a weight-1 $X$-type error. Therefore, if the distance of $C_z$ is $d_z\geq 3$ and the code $C_x$ can distinguish all errors in the the form \ref{form_b}, then $\mathcal{E}_-$ is distinguishable by $C$. The same argument can also be applied to circuits for measuring $X$ stabilizers. 

We can see that the ability of the code to distinguish errors of the form \ref{form_b} is crucial in a flag-FTEC protocol. In order to develop a flag-FTEC protocol for cyclic CSS codes, the following definitions will be very useful:

\begin{definition}{\underline{\smash{Left cyclic-shift}}}
	
	Let $P = P_1\otimes \cdots \otimes P_n$ be an $n$-qubit Pauli operator and $l \in \{0,1,\dots,n-1\}$. The \emph{$l$-qubit left cyclic-shift} of the operator P, denoted by $\mathcal{L}(P,l)$, is defined as
	\begin{align}
	\mathcal{L}(P, 0) &= P, \\
	\mathcal{L}(P, l) &= P_{l+1}\otimes \cdots \otimes P_n \otimes P_1 \otimes \cdots \otimes P_l \quad \text{for}\quad l\neq 0.
	\end{align}
	\label{def:Leftshift}%
\end{definition}

\begin{definition}{\underline{\smash{Consecutive error set}}}
	
	Let $n$ be the number of qubits and $l \in \{0,1,\dots,n-1\}$. A \emph{consecutive-$X$ error set} $\mathcal{E}^x_{l,n}$ and a \emph{consecutive-$Z$ error set} $\mathcal{E}^z_{l,n}$ are sets of the form
	\begin{align}
	\mathcal{E}^x_{l,n} &=  \{\mathcal{L}(I^{\otimes n-p} \otimes X^{\otimes p},l):\ p \in \{0,1,\dots,n-1\} \}, \\
	\mathcal{E}^z_{l,n} &=  \{\mathcal{L}(I^{\otimes n-p} \otimes Z^{\otimes p},l):\ p \in \{0,1,\dots,n-1\} \}.
	\end{align}
	A \emph{consecutive error product set} $\mathcal{E}^P_{l,n}$ is defined as
	\begin{align}
	\mathcal{E}^P_{l,n} &=  \{E_x \cdot E_z: E_x \in \mathcal{E}^x_{l,n}, E_z \in \mathcal{E}^z_{l,n} \}.
	\end{align}
	\label{def:Consecutive}%
\end{definition}

In order to distinguish all errors in each consecutive error set, the $X$ and $Z$ stabilizer matrices must satisfy the conditions in the following lemma:

\begin{lemma}
	
	Let $C$ be a CSS code constructed from the classical cyclic codes $C_x$ and $C_z$ following \cref{Theorem:CSS} with parity check matrices $H_x$ and $H_z$ of the form
	\begin{align}
	H_x =\begin{pmatrix}
	x_{1,1} & x_{1,2} & \dots & x_{1,n} \\
	x_{2,1} & x_{2,2} & \dots & x_{2,n} \\
	\dots & & & \dots \\
	x_{r_x,1} & x_{r_x,2} & \dots & x_{r_x,n} \\
	\end{pmatrix},
	\end{align}
	\begin{align}
	H_z =\begin{pmatrix}
	z_{1,1} & z_{1,2} & \dots & z_{1,n} \\
	z_{2,1} & z_{2,2} & \dots & z_{2,n} \\
	\dots & & & \dots \\
	z_{r_z,1} & z_{r_z,2} & \dots & z_{r_z,n} \\
	\end{pmatrix},
	\end{align}
	and let $\mathcal{E}^x_{0,n}$, $\mathcal{E}^z_{0,n}$, and $\mathcal{E}^P_{0,n}$ be consecutive-$X$ error set, consecutive-$Z$ error set, and consecutive error product set, respectively. Then, 
	\begin{enumerate}
		\item $\mathcal{E}^z_{0,n}$ is distinguishable by $C$ iff for all $p,q\in\{0,...,n{-}1\}$ such that $p>q$, there exists $i\in\{1,\dots,r_x\}$ such that $x_{i,n-p+1}\oplus\dots\oplus x_{i,n-q}=1$.
		\item $\mathcal{E}^x_{0,n}$ is distinguishable by $C$ iff for all $p,q\in\{0,...,n{-}1\}$ such that $p>q$, there exists $i\in\{1,\dots,r_z\}$ such that
		$z_{i,n-p+1}\oplus\dots\oplus z_{i,n-q}=1$.
		\item $\mathcal{E}^P_{0,n}$ is distinguishable by $C$ iff both $\mathcal{E}^z_{0,n}$ and $\mathcal{E}^x_{0,n}$ are distinguishable by $C$.
	\end{enumerate}
	\label{Lemma:Lem1}%
\end{lemma}

A proof of \cref{Lemma:Lem1} is given in \cref{app:LemmaProofs}.

Note that consecutive error sets in \cref{Lemma:Lem1} are defined on $n$ qubits. In particular, consecutive error sets defined on $m$ qubits are distinguishable iff the submatrices of $H_x$ and $H_z$ corresponding to measurements on $m$ qubits satisfy similar conditions. In the next section, we will show that the cyclic symmetry of cyclic CSS codes can simplify the conditions in \cref{Lemma:Lem1}.

\section{Cyclic CSS codes and error distinguishability}
\label{sec:CyclicCodesDistinguishability}%
In \cref{sec:ReviewFlag}, the conditions for distinguishing errors in the consecutive error sets are given in \cref{Lemma:Lem1}. Notice that there are some sufficient conditions for distinguishability in  Statements 1 and 2 that are similar, different only by some qubit shift. It is possible to simplify \cref{Lemma:Lem1} if the CSS code has cyclic symmetry. In this section, we begin by stating the definition of classical cyclic codes and outlining some of their properties. Afterwards, \cref{Lemma:Cyclic_gen,Lemma:Lem3} will be given, and \cref{theo:MainTheorem} which is the main theorem in this work will be proved.

\begin{definition}{\underline{\smash{Classical cyclic code}} \cite{MS77}}

Let $C$ be a classical binary linear code of length $n$. $C$ is \emph{cyclic} if any cyclic shift of a codeword is also a codeword, i.e., if $(c_1,c_2,\dots,c_n)$ is in a codeword, then so is  $(c_n,c_1,\dots,c_{n-1})$. 
\end{definition}

Let $C$ be a classical cyclic code of length $n$. There exists a unique generator polynomial $g(x)=\sum_{i=0}^\alpha g_ix^i$ which is also a unique monic polynomial of minimal degree in $C$ such that $C$ is generated by the generator matrix
\begin{align}
\begin{pmatrix}
g_0 & g_1 & g_2 & \dots & g_\alpha & 0 & \dots & 0 \\
0 & g_0 & g_1 & \dots & g_{\alpha-1} & g_\alpha & \dots & 0 \\
\dots & & & & & & & \dots\\
0 & \dots & g_0 & \dots & & & \dots & g_\alpha
\end{pmatrix}.
\end{align}
The polynomial $h(x) = (x^n-1)/g(x) = \sum_{i=0}^\beta h_ix^i$ is called the check polynomial of $C$. The parity check matrix of $C$ is
\begin{align}
\begin{pmatrix}
h_\beta & h_{\beta-1} & \dots & h_1 & h_0 & 0 & \dots & 0 \\
0 & h_\beta & \dots & h_2 & h_1 & h_0 & \dots & 0 \\
\dots & & & & & & & \dots\\
0 & \dots & h_\beta & \dots & & & \dots & h_0
\end{pmatrix}.
\end{align}

It is known that any classical Hamming code can be made cyclic \cite{MS77}. Thus, a cyclic CSS code can be constructed from permuting columns of a quantum Hamming code's stabilizer matrices. In \cite{QuantumCodesCyclic}, it was shown how to construct a cyclic CSS code from two classical cyclic codes.

By the symmetries of a cyclic code, we can show in the following lemma that the left cyclic-shift of operators in the generating set also generates the same code.

\begin{lemma}
	
	Let $C$ be a CSS code constructed from the classical cyclic codes $C_x$ and $C_z$ following \cref{Theorem:CSS}. Suppose that the stabilizer group of $C$ can be generated by $\{g_1,g_2,\dots,g_{n-k}\}$, then the stabilizer group of $C$ can also be generated by $\{\mathcal{L}(g_{1},l),\mathcal{L}(g_{2},l),\dots,\mathcal{L}(g_{n-k},l)\}$ for any $l \in \{0,1,\dots,n-1\}$.
	\label{Lemma:Cyclic_gen}%
	
\end{lemma}
The proof of \cref{Lemma:Cyclic_gen} is given in \cref{app:LemmaProofs}. 

In the previous section, \cref{Lemma:Lem1} gives sufficient and necessary conditions for a CSS code to be able to distinguish all errors in the consecutive error sets. The conditions can be simplified by using the symmetry of cyclic codes as follows:

\begin{lemma}
	
	Let $C$ be a CSS code constructed from the classical cyclic codes $C_x$ and $C_z$ (following \cref{Theorem:CSS}) with parity check matrices $H_x$ and $H_z$,
	\begin{align}
	H_x &=\begin{pmatrix}
	x_{1,1} & x_{1,2} & \dots & x_{1,n} \\
	x_{2,1} & x_{2,2} & \dots & x_{2,n} \\
	\dots & & & \dots \\
	x_{r_x,1} & x_{r_x,2} & \dots & x_{r_x,n} \\
	\end{pmatrix},\\
	H_z &=\begin{pmatrix}
	z_{1,1} & z_{1,2} & \dots & z_{1,n} \\
	z_{2,1} & z_{2,2} & \dots & z_{2,n} \\
	\dots & & & \dots \\
	z_{r_z,1} & z_{r_z,2} & \dots & z_{r_z,n} \\
	\end{pmatrix}.
	\end{align}
	Let $l \in \{0,1,\dots,n{-}1\}$, and let $\mathcal{E}^x_{l,n}$, $\mathcal{E}^z_{l,n}$, and $\mathcal{E}^P_{l,n}$ be consecutive-$X$ error set, consecutive-$Z$ error set, and consecutive error product set, respectively. Then, 
	\begin{enumerate}
		\item $\mathcal{E}^z_{l,n}$ is distinguishable by $C$ iff for all $u_x\in\{2,\dots,n\}$, there exists $i\in\{1,\dots,r_x\}$ such that $x_{i,u_x}\oplus\dots\oplus x_{i,n}=1$.
		\item $\mathcal{E}^x_{l,n}$ is distinguishable by $C$ iff for all $u_z\in\{2,\dots,n\}$, there exists $i\in\{1,\dots,r_z\}$ such that $z_{i,u_z}\oplus\dots\oplus z_{i,n}=1$.
		\item $\mathcal{E}^P_{l,n}$ is distinguishable by $C$ iff both $\mathcal{E}^z_{l,n}$ and $\mathcal{E}^x_{l,n}$ are distinguishable by $C$.
	\end{enumerate}
	\label{Lemma:Lem3}%
\end{lemma}

The cyclic symmetry of $C_x$ and $C_z$ can simplify \cref{Lemma:Lem1}, resulting in fewer sufficient conditions for distinguishability in Statements 1 and 2; we can fix $q$ in \cref{Lemma:Lem1} to be $0$ and choose $u = n-p+1$. This reduces the number of error pairs in consecutive error sets to be distinguished. Moreover, all statements in \cref{Lemma:Lem3} can also be applied to consecutive error sets of any $l \in \{0,1,\dots,n-1\}$. A proof of \cref{Lemma:Lem3} is given in \cref{app:LemmaProofs}.

Now we are ready to prove a main theorem in this work.

\begin{theorem}
	
	Let $C$ be an \codepar{n,k,d} CSS code constructed from the $[n,k_x,d_x]$ classical cyclic code $C_x$ and the $[n,k_z,d_z]$ classical cyclic code $C_z$, $l \in \{0,1,\dots,n-1\}$, and $\mathcal{E}^P_{l,n}$ be a consecutive error product set. If both $d_x, d_z \ge 3$, then $\mathcal{E}^P_{l,n}$ is distinguishable by $C$.
	\label{theo:MainTheorem}%
\end{theorem}
	
\begin{proof}
	\;
	
	Suppose by contradiction that $\mathcal{E}_{l,n}^P$ is not distinguishable by $C$. Then, at least one of $\mathcal{E}_{l,n}^z$ and $\mathcal{E}_{l,n}^x$ is not distinguishable by $C$. Similar analysis applies to either case, so, suppose $\mathcal{E}_{l,n}^z$ is not distinguishable by $C$. We next invoke \cref{Lemma:Lem3}, and to do so, let the cyclic code $C_x$ has parity check matrix 
	\begin{equation}
	H_x =\begin{pmatrix}
	x_{1,1} & x_{1,2} & \dots & x_{1,n} \\
	x_{2,1} & x_{2,2} & \dots & x_{2,n} \\
	\dots & & & \dots \\
	x_{r_x,1} & x_{r_x,2} & \dots & x_{r_x,n} \\
	\end{pmatrix},
	\end{equation}
	where $x_{i_1+1,(j+1)\pmod n}=x_{i_1,j}$ for all $i_1\in\{1,\dots,r_x-1\}$ and $j\in\{1,\dots,n\}$. 
	%
	%
	By Statement 1 of \cref{Lemma:Lem3}, we know that $\mathcal{E}_{l,n}^z$ is indistinguishable by $C$ iff there exists $u_x \in\{2,3,...,n\}$ such that $x_{i,u_x}\oplus\dots\oplus x_{i,n}=0$ for all 
	$i \in \{1, \cdots, r_x \}$; i.e., there exists a pair of errors in $\mathcal{E}_{l,n}^z$ which cannot be distinguished by any generator of $C$.
	%
	%
	For $i=1$, we have 
	\begin{align}
	x_{1,u_x}\oplus\dots\oplus x_{1,n}&=0 \,.
	\label{eq:Cond1}%
	\end{align}
	Using \cref{Lemma:Cyclic_gen}, we obtain a new generating set for $C$ where each generator is the 1-qubit left cyclic-shift of the old one.  
	Applying the above to the new generator set, we have
	\begin{align}
	x_{1,u_x+1}\oplus\dots\oplus x_{1,n}\oplus x_{1,1}=0. \label{eq:Cond2}
	\end{align}
	Now repeating the left cyclic shifts gives 
	\begin{align}
	x_{1,u_x+2}\oplus\dots\oplus x_{1,n}\oplus x_{1,1}\oplus x_{1,2}&=0, \label{eq:Cond3} \\
	&\vdots\nonumber\\
	x_{1,u_x-1}\oplus x_{1,1}\oplus\dots\oplus x_{1,n-1}&=0 \,. \nonumber
	\end{align}
	From \cref{eq:Cond1,eq:Cond2}, $x_{1,1}=x_{1,u_x}$; from \cref{eq:Cond2,eq:Cond3}, $x_{1,2}=x_{1,u_x+1}$, and so on, until we obtain $x_{1,n}=x_{1,u_x+(n{-}1) \pmod n}$ (in other words, $x_{1,j}=x_{1,(u_x-1+j)\pmod n}$ for all $j\in\{1,\dots,n\}$).  
	Let $w_x = \text{GCD}(u_x{-}1,n)$, the greatest common divisor of $u_x{-}1$ and $n$. The conditions become
	\begin{align}
	x_{1,j}=x_{1,j+w_x}=x_{1,j+2w_x}=...=x_{1,j+n-w_x},
	\end{align}
	for all $j \in\{1,\dots,w_x\}$. Repeating the above steps for all $i$, we obtain
	\begin{align}
	x_{i,j}=x_{i,j+w_x}=x_{i,j+2w_x}=...=x_{i,j+n-w_x},
	\end{align}
	for all $i \in\{1,\dots,r_x\}$, $j\in\{1,\dots,w_x\}$.
	
	From the above, we see that any error of the form $Z_{l_x}Z_{l_x+w_x}$ (where $l_x \in \{ 1, \cdots, n-w_x \}$) commutes with all stabilizer generators. Now let us consider two cases:
	\begin{itemize}
		\item \textbf{Case 1:} At least one operator of the form  $Z_{l_x}Z_{l_x+w_x}$ is not in the stabilizer.
		
		In this case, the distance $d$ of the code $C$ is at most two. Since $d \geq \min\{d_x,d_z\}$ (see \cref{Theorem:CSS}), this contradicts our assumption that both $d_x,d_z \geq 3$.
		\item  \textbf{Case 2:} All operators of the form $Z_{l_x}Z_{l_x+w_x}$ are in the stabilizer.
		
		In this case,  there exists a set of coefficients $a_1, \cdots, a_{r_z} \in \{0,1 \}$ such that $(g^{z}_1)^{a_1}\cdots (g^{z}_{r_z})^{a_{r_z}} = Z_{l_x}Z_{l_x+w_x}$, where $g^z_i$ is the $Z$-type generator corresponding to the $i\textsuperscript{th}$ row of $H_{z}$. This means that the $Z$-part of $\sigma(Z_{l_x}Z_{l_x+w_x})$ is a codeword in $C_{z}^{\perp}$. Since $C_{z}^{\perp} \subseteq C_{x}$ by the construction of CSS codes, we have that the $Z$-part of $\sigma(Z_{l_x}Z_{l_x+w_x})$ is a codeword in $C_{x}$. Because the distance of classical codes is given by the minimum Hamming weight of the codewords, we have that $d_{x} \le 2$ which contradicts our assumption that $d_{x} \ge 3$.
	\end{itemize}
	
	
\end{proof}

Although consecutive error product set $\mathcal{E}_{l,n}^P$ is distinguishable by any cyclic CSS code satisfying \cref{theo:MainTheorem}, we cannot  construct an FTEC protocol using the circuit in \cref{fig:StabFTwithAncilla} directly since the possible errors might not be in the consecutive form without qubit permutation. Moreover, permuting qubits will break the cyclic symmetry and $\mathcal{E}_{l,n}^P$ might no longer be distinguishable. In the next section, we will use \cref{theo:MainTheorem} to find a 1-flag circuit for distance-3 cyclic CSS codes that can be used in a fault-tolerant protocol satisfying \cref{Def:FaultTolerantDef}. We point out that since $p,q$ in \cref{Lemma:Lem1} are chosen to be in the set $\{ 0, \cdots, n-1 \}$, if a cyclic CSS code 
can correct errors of weight $\leq t$, then the flag circuits should be designed such that if there are $\leq t$ faults during the FTEC protocol, an error of weight $n$ cannot occur.

\section{Fault-tolerant error correction protocol for distance-3 cyclic CSS codes}
\label{sec:Protocol}%
\begin{figure*}
	\centering
	\includegraphics[width=1.0\textwidth]{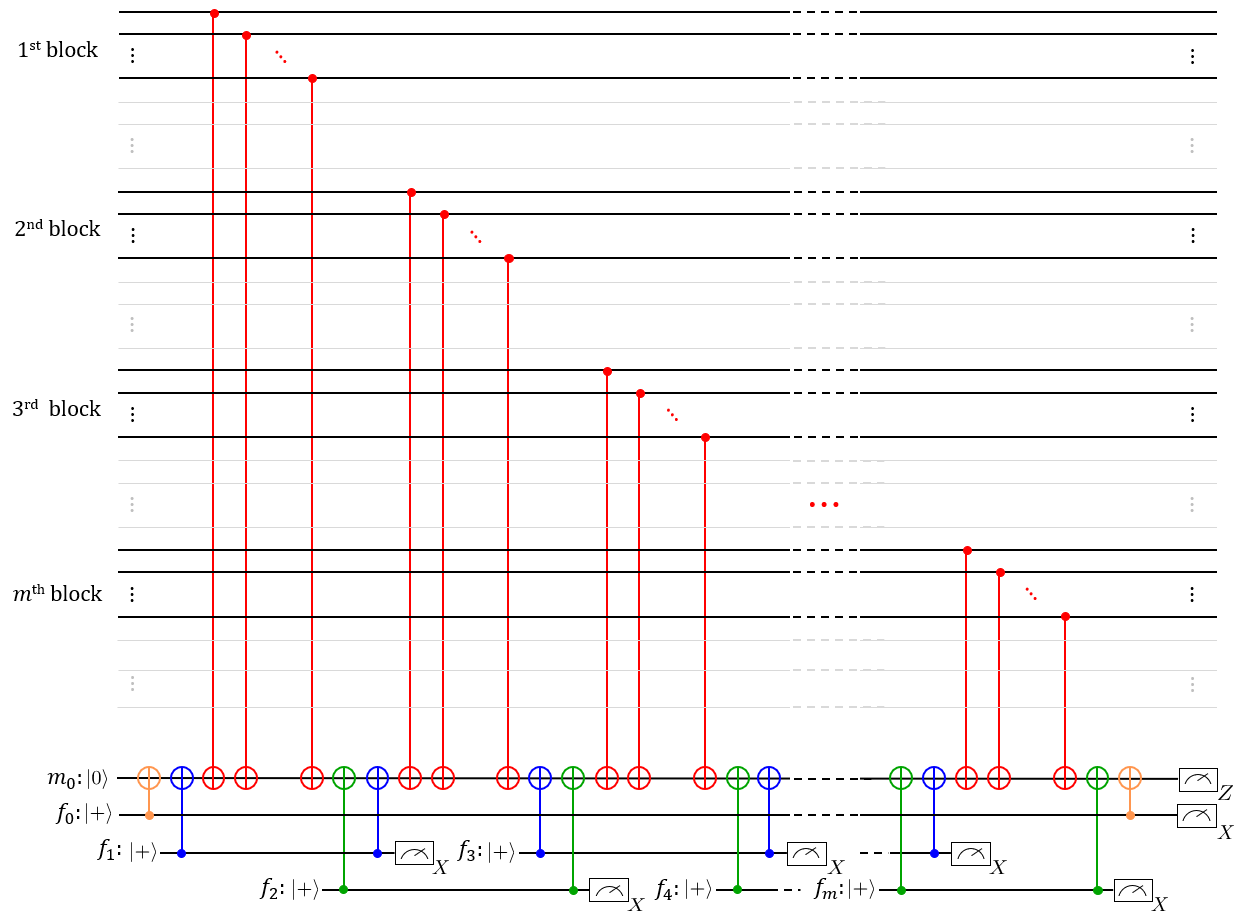}
	\caption{Illustration of a 1-flag circuit applicable to distance-3 cyclic CSS codes. The circuit measures stabilizers of the form $Z^{\otimes a_1}\otimes I^{\otimes b_1}\otimes Z^{\otimes a_2}\otimes I^{\otimes b_2}\otimes \cdots \otimes Z^{\otimes a_m}\otimes I^{\otimes b_m}$. The flag qubits are represented by the labels $f_1, \cdots , f_m$. Information from the flag outcomes along with the protocol given in \cref{sec:Protocol} enable the construction of a flag-FTEC protocol which satisfies \cref{Def:FaultTolerantDef}. (For grayscale version, \emph{red} CNOT gates are CNOT gates connecting between a data qubit and qubit $m_0$. The \emph{orange}, \emph{blue}, and \emph{green} CNOT gates have control qubits $f_0$, $f_i$ for odd $i$, and $f_i$ for even $i$, respectively.)}
	\label{fig:GeneralWflagFig}%
\end{figure*}

Fault-tolerant error correction is one of the most important building blocks for fault-tolerant quantum computation. In this section, a flag-FTEC protocol for distance-3 cyclic CSS codes is developed \footnote{Note that our protocol and circuit can also be applied to higher distance codes if we only consider correcting errors introduced by at most one fault.}. A 1-flag circuit for cyclic CSS codes of distance 3 which is required for the flag-FTEC protocol is provided in \cref{fig:GeneralWflagFig} (see \cref{Def:tFlaggedCircuitDef} for the definition of a $t$-flag circuit). Here we adapt the idea of localizing circuit faults from \cite{CR17v1}.

Suppose that the stabilizer generator being measured is of the form 
\begin{align}
P = Z^{\otimes a_1}\otimes I^{\otimes b_1}\otimes Z^{\otimes a_2}\otimes I^{\otimes b_2}\otimes \cdots \otimes Z^{\otimes a_m}\otimes I^{\otimes b_m}, \nonumber
\end{align} 
where $a_i>0$ and $b_i \geq 0$ are integers. The $i^\text{th}$ sub-block consists of $a_i$ qubits, which are from the $\sum_{j=1}^{i-1}(a_j+b_j)+1$'th qubit to the $\sum_{j=1}^{i-1}(a_j+b_j)+a_i$'th qubit.

Notice that the blue, green and orange CNOT gates in the circuit of \cref{fig:GeneralWflagFig} always come in pairs. This is to ensure that when fault-free, the circuit implements a projective measurement of the stabilizer without flagging. In what follows, we will refer to the first blue, green or orange CNOT of a pair as an \textit{opening} CNOT and the second blue, green or orange CNOT as a \textit{closing} CNOT. Given these definitions, we have the following claim:

\begin{claim}
During the measurement of $P = Z^{\otimes a_1}\otimes I^{\otimes b_1}\otimes Z^{\otimes a_2}\otimes I^{\otimes b_2}\otimes \cdots \otimes Z^{\otimes a_m}\otimes I^{\otimes b_m}$ using the circuit in \cref{fig:GeneralWflagFig}, the following can occur:
\begin{enumerate}
    
     \item If there are no faults, none of the $f_i$ ancilla qubits will flag. 
     \item  A fault at a CNOT location resulting in a $ZZ$ error is equivalent to the prior CNOT failing resulting in an $IZ$ error (here $Z$ acts on the target qubit).
     \item  Suppose that a fault occurs on one of the red CNOT's and causes a $Z$ error on the ancilla $m_0$. If the fault occurs on sub-block $a_i$ where $i\geq 1$, only the ancillas $f_0$ and $f_i$ will flag.
     \item  Suppose that a fault occurs on a blue or green CNOT. Let the control qubit be the ancilla $f_i$. If it is the opening CNOT and causes a $Z$ error on ancilla $m_0$, the ancillas $f_0$, $f_i$, and $f_{i-1}$ will flag. If it is the closing CNOT and causes a $Z$ error on the ancilla $m_0$, the ancillas $f_0$ and $f_{i+1}$ will flag. However if the fault occurs on a blue or green CNOT's at the boundary \footnote{By boundary we are referring to either the first blue CNOT after the $1^\text{th}$ sub-block or the last green CNOT after the $m^\text{th}$ sub-block.}, if the opening CNOT of $f_1$ is faulty, $f_0$ and $f_1$ will flag, and if the closing CNOT of $f_m$ is faulty, only $f_0$ will flag.
     \item  A fault occurring at an orange CNOT gate will not cause a data qubit error (since a $Z$ spreading to all qubits is equivalent to the stabilizer being measured). Furthermore, only the ancilla $f_0$ can flag in this case (depending on whether the error was of the form $IZ$ or $ZZ$ and also whether it occurred on the opening or closing orange CNOT).    
\end{enumerate}
\label{cl:PrincipalClaim}%
\end{claim}

From the above claim, one can verify that a single fault resulting in a data qubit error $E$ with $\text{min}_Q(\text{wt}(E Q)) > 1$ where $Q \in \langle P, g_1, \cdots, g_{n-k} \rangle$ will always cause at least one flag qubit to flag (see \cref{Def:tFlaggedCircuitDef}). Thus the circuit in \cref{fig:GeneralWflagFig} is a 1-flag circuit. Note that an analogous claim can be made for $X$-type stabilizers.

Before describing the FTEC protocol, we require one more definition:
\begin{definition}{\underline{\smash{Minimum weight correction}}}

Given the syndrome $s=s(E)$ of an error $E$, we let $E_{\text{min}}(s)$ be a \emph{minimal weight correction} of $E$. 
\label{def:MinErr}%
\end{definition}

Note that many errors can lead to the same syndrome. In particular, errors corresponding to the same syndrome differ by some multiplication of stabilizers or logical operators. If error $E$ is correctable, applying $E_{\text{min}}(s)$ can correct such error as we expected. However, if $E$ is not correctable (i.e., $\text{wt}(E) > t$), applying $E_{\text{min}}(s)$ will project the codeword back to the coding subspace, but the resulting codeword may differ from the original codeword by some logical operation. This property of minimal weight correction is required so that the FTEC protocol satisfy the second FTEC condition in \cref{Def:FaultTolerantDef}.

Using \cref{theo:MainTheorem,cl:PrincipalClaim,def:MinErr}, we now describe a FTEC protocol that satisfies \cref{Def:FaultTolerantDef} for distance-3 cyclic CSS codes using a procedure adapted from \cite{CB17}. In what follows, we define $s^{(r)} = (s^{(r)}_x|s^{(r)}_z)$ to be the syndrome obtained during round $r$ (either using flag or non-flag circuits), where $s^{(r)}_x$ and $s^{(r)}_z$ are the syndromes obtained from $X$-type and $Z$-type stabilizers, respectively.

\textbf{FTEC Protocol:} 

Let $C$ be an \codepar{n,k,d} cyclic CSS code satisfying \cref{theo:MainTheorem} with stabilizer group $S = \langle g_1 , \cdots , g_{n-k} \rangle$. Let ${\cal C}(g_i)$ be the 1-flag circuit of \cref{fig:GeneralWflagFig} for stabilizer $g_i$. Repeat the syndrome measurement (measurement of all stabilizer generators) using the 1-flag circuits until one of the following is satisfied:
\begin{enumerate}
        \item If the syndrome is repeated twice in a row and there are no flags, apply $E_{\text{min}}(s^{(1)})$.
	\item If there are no flags and the syndromes $s^{(1)}$ and $s^{(2)}$ differ, repeat the syndrome measurement using non-flagged circuits. Apply the correction $E_{\text{min}}(s^{(3)})$.
        \item If $f_0$ does not flag but $f_i$ flags (with $i \ge 1$) during round one, stop. Repeat the syndrome measurement using non-flag circuits and apply $E_{\text{min}}(s^{(2)})$. If there are no flags in the first round but in round two $f_i$ flags and $f_0$ does not flag, stop. Apply  $E_{\text{min}}(s^{(1)})$.	
	\item If $f_0$ flags at round $r$ anytime during the protocol, stop and do one of the following:
	\begin{enumerate}
		\item If $f_i$ does not flag for all $i\geq1$, repeat the syndrome measurement using non-flag circuits. Apply $E_{\text{min}}(s^{(r+1)})$.
		\item If there is only one $i$ such that $f_i$ flags (with $i \ge 1$), apply $I^{\otimes c} \otimes Z^{\otimes a_{i+1}} \otimes I^{\otimes b_{i+1}} \otimes \cdots \otimes Z^{\otimes a_m} \otimes I^{\otimes b_m}$ to the data if the stabilizer being measured is a $Z$ stabilizer or $I^{\otimes c} \otimes X^{\otimes a_{i+1}} \otimes I^{\otimes b_{i+1}} \otimes \cdots \otimes X^{\otimes a_m} \otimes I^{\otimes b_m}$ if it is an $X$ stabilizer, where $c=\sum_{j=1}^i (a_j+b_j)$. Repeat the syndrome measurement using non-flag circuits yielding syndrome $s^{(r+1)} = (s^{(r+1)}_x|s^{(r+1)}_z)$. 
		\begin{enumerate}
			\item If the stabilizer being measured is a $Z$ stabilizer and there is an element $E_z$ in $\mathcal{E}^{z}_{l,n}$ where $l=n-c+b_i$ that satisfies $s(E_z)=s^{(r+1)}_x$, apply $E_z$ followed by $E_{\text{min}}(s^{(r+1)}_z)$. Otherwise, apply  $E_{\text{min}}(s^{(r+1)})$.
			\item If the stabilizer being measured is an $X$ stabilizer and there is an element $E_x$ in $\mathcal{E}^{x}_{l,n}$ where $l=n-c+b_i$ that satisfies $s(E_x)=s^{(r+1)}_z$, apply $E_x$ followed by $E_{\text{min}}(s^{(r+1)}_x)$. Otherwise, apply  $E_{\text{min}}(s^{(r+1)})$.
		\end{enumerate}
		\item If there is an $i$ such that $f_i$ and $f_{i+1}$ flag, perform the same sequence operations as in 4(b).
	\end{enumerate}	
\end{enumerate}

To see that the above protocol satisfies \cref{Def:FaultTolerantDef}, we will assume that there is at most one fault during the protocol. If a fault in any of the CNOT gates introduces a $Z$ error on ancilla $m_0$, then $f_0$ and at least one $f_i$ (with $i \geq 1$) will flag (unless the first orange CNOT introduces an error of the form $ZZ$ or the last orange CNOT introduces an error of the form $IZ$ which in both cases, there will be no data qubit error). If there is only one flag during round one, either $f_0$ or $f_i$, then the fault could either have been caused by a measurement error, idle qubit error on the ancilla $f_0$ or $f_i$, or an error on the control qubit of the CNOT gate interacting with $f_0$ or $f_i$. However in all three cases, the error could not have spread to the data. By repeating the syndrome measurement and applying $E_{\text{min}}(s^{(2)})$, both criteria of \cref{Def:FaultTolerantDef} will be satisfied. Note that if $f_i$ flags during round two, then the syndrome obtained during round one corresponds to the data qubit error (since there could not have been a measurement error giving the wrong syndrome during the first round), so correcting using $s^{(1)}$ will again satisfy \cref{Def:FaultTolerantDef}.

Next, let us consider the case where none of the $f_i$ ancillas flag. By the circuit construction, a single fault can introduce an error $E$ with $\text{wt}(E) \le 1$. If the same syndrome is repeated twice in a row, then applying $E_{\text{min}}(s^{(1)})$ can result in a data error of weight at most one. If $s^{(1)} \neq s^{(2)}$, then a fault occurred in either the first or second round. Thus repeating the syndrome measurement a third time and applying $E_{\text{min}}(s^{(3)})$ will remove the data errors or project the code back to the coding subspace. 

Next we consider the case where a fault happens on a red CNOT introducing a $Z$ error on the ancilla $m_0$ and a $P$ error on the data qubit where $P \in \{I,X,Y,Z\}$. If the fault occurs on the $i^\text{th}$ sub-block, then $f_0$ will flag and there will be only one $i \ge 1$ such that $f_i$ flags. Applying $I^{\otimes c} \otimes Z^{\otimes a_{i+1}} \otimes I^{\otimes b_{i+1}} \otimes \cdots \otimes Z^{\otimes a_m} \otimes I^{\otimes b_m}$ where $c=\sum_{j=1}^i (a_j+b_j)$ to the data if the stabilizer being measured is a $Z$ stabilizer (or $I^{\otimes c} \otimes X^{\otimes a_{i+1}} \otimes I^{\otimes b_{i+1}} \otimes \cdots \otimes X^{\otimes a_m} \otimes I^{\otimes b_m}$ if it is an $X$ stabilizer) guarantees that the resulting error is a product of $Z$-type error from $\mathcal{E}^{z}_{l,n}$ and an $X$-type error of weight at most 1 (or a product of $X$-type error from $\mathcal{E}^{x}_{l,n}$ and a $Z$-type error of weight at most 1). By \cref{theo:MainTheorem}, errors in the set $\mathcal{E}^{z}_{l,n}$ (or  $\mathcal{E}^{x}_{l,n}$) can be distinguished. Thus applying the correction in 4(b) of the protocol will remove the error if there are no input errors. However, if there is an input error, then applying $E_{\text{min}}(s^{(r+1)})$ will project the code back to the coding subspace.

Lastly, if a fault occurs on a blue or green CNOT, then from \cref{cl:PrincipalClaim} either the case in 4(b) or 4(c) will be satisfied. However in both cases, the $Z$ error will spread to the data in the same way. Hence the correction proposed in 4(c) will satisfy the fault-tolerance criteria of \cref{Def:FaultTolerantDef}.

A list of possible faults during the flag-FTEC protocol and corresponding correction procedures is given in \cref{tab:FTEC_fault} in \cref{app:FaultTable}.

\section{Fault-tolerant measurement protocol for distance-3 cyclic CSS codes} 
\label{sec:MeasProtocol}%
Besides FTEC, there are other important components for fault-tolerant computation: FT state preparation, FT measurement, and FT quantum gate implementation. In this section, we provide a flag-FT measurement protocol for distance-3 cyclic CSS codes. The measurement protocol plays an important role in fault-tolerant quantum computation on cyclic CSS codes since it can also be used as a subroutine for FT state preparation, FT quantum gate implementation, and other techniques, described later in \cref{sec:Applications}.

The flag-FT protocol provided in this section is similar to the flag-FTEC protocol in \cref{sec:Protocol} except that the idea of consecutive error correction is developed so that it is applicable not only to stabilizer measurements but also to measurements of any Pauli operator commuting with all generators. We begin by introducing the definition of fault-tolerant non-destructive measurement adapted from \cite{AGP06} as follows:
\begin{definition}{\underline{\smash{Fault-tolerant non-destructive measure-}} \underline{\smash{ment}}}
	
	For $t = \lfloor (d-1)/2\rfloor$, a non-destructive measurement protocol using a distance-$d$ stabilizer code $C$ is $t$-\emph{fault-tolerant} if the following two conditions are satisfied:
	\begin{enumerate}
		\item For an input codeword with error of weight $v_{1}$, if $v_{2}$ faults occur during the measurement protocol with $v_{1} + v_{2} \le t$, ideally decoding the output state after measurement gives the same state as ideally decoding the input state and then performing ideal non-destructive measurement. The result obtained from measuring the input codeword is the same as that of an ideal measurement on the ideally-decoded input state.
		\item For an input codeword with error of weight $v_{1}$, if $v_{2}$ faults occur during the measurement protocol with $v_{1} + v_{2} \le t$, the output state differs from a codeword by an error of at most weight $v_{1}+v_{2}$.
	\end{enumerate}
	\label{Def:FTmeasDef}%
\end{definition}
\noindent (Here we need to modify \cref{Def:FTmeasDef} from the usual definition of fault-tolerant (destructive) measurement since we would like to obtain both measurement result and post-measurement state. These ingredients are important in the applications discussed in \cref{sec:Applications}.)

Suppose that the operator being measured $P$ commutes with all generators and is of the form
\begin{align}
P = P_1^{\otimes a_1} \otimes I^{\otimes b_1} \otimes P_2^{\otimes a_2} \otimes I^{\otimes b_2} \otimes \cdots \otimes P_m^{\otimes a_m} \otimes I^{\otimes b_m},
\end{align}
where $a_i>0$ and $b_i \geq 0$ are integers and $P_i \in \{X,Y,Z\}$. The $i^\text{th}$ sub-block consists of $a_i$ qubits acted on by $P_i^{\otimes a_i}$. A 1-flag circuit for operator measurement is similar to the circuit given in \cref{fig:GeneralWflagFig}, while the measurements of $Z$, $X$, and $Y$ operators correspond to CNOT gates, the gates shown in \cref{fig:XNOTgates}, and the gates shown in \cref{fig:YNOTgates}. With the slight modification where CNOT gates are replaced by the gates for measuring $P_i \in \{X,Y,Z\}$, one can verify that \cref{cl:PrincipalClaim} is also applicable in this setting.

\begin{figure}[htbp]
	\centering
	\includegraphics[width=0.3\textwidth]{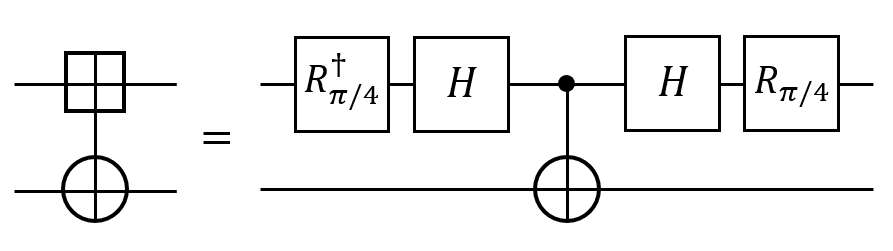}
	\caption{Quantum gates for measuring $Y$ operator where $R_{\pi/4} = \text{diag}(1,i)$.}
	\label{fig:YNOTgates}%
\end{figure}

Using \cref{theo:MainTheorem} and \cref{cl:PrincipalClaim}, we now describe a flag-FT measurement protocol that satisfies \cref{Def:FTmeasDef} for distance-3 cyclic CSS codes. Here we define $m^{(r_1)}$ to be the measurement result obtained from operator measurement (using either flag or non-flag circuits) during round $r_1$, and define $s^{(r_2)}=(s^{(r_2)}_x|s^{(r_2)}_z)$ to be the syndrome obtained from syndrome measurement (using either flag or non-flag circuits) during round $r_2$ of error correction. The protocol is as follows:

\textbf{Flag-FT Operator Measurement Protocol:}

Let $C$ be an \codepar{n,k,d} cyclic CSS code satisfying \cref{theo:MainTheorem}. Let ${\cal C}(P)$ be the 1-flag circuit of \cref{fig:GeneralWflagFig} for measuring a Pauli operator $P$ of the form
\begin{align}
P = P_1^{\otimes a_1} \otimes I^{\otimes b_1} \otimes P_2^{\otimes a_2} \otimes I^{\otimes b_2} \otimes \cdots \otimes P_m^{\otimes a_m} \otimes I^{\otimes b_m},
\end{align} 
where $P$ commutes with all generators of $C$. Repeat the operator measurement using the 1-flag circuits until one of the following is satisfied:

\begin{enumerate}
	\item If first two operator measurement results coincide $(m^{(1)} = m^{(2)})$ and there are no flags, perform the syndrome measurement twice using 1-flag circuits for flag-FTEC.
	\begin{enumerate}
		\item If $s^{(1)} = s^{(2)} = 0$ and there are no flags during both syndrome measurement rounds, output $m^{(1)}$.
		
		\item If $s^{(1)} \neq s^{(2)}$ or at least one flag qubit flags during the syndrome measurement, apply the correction described by the flag-FTEC protocol of \cref{sec:Protocol}, then output $m^{(1)}$.
		
		\item If $s^{(1)} = s^{(2)} \neq 0$ and there are no flags during the syndrome measurement, apply the correction $E_\text{min}(s^{(1)})$. Repeat the operator measurement using a non-flag circuit, then output $m^{(3)}$. 
	\end{enumerate}
	
	\item If $m^{(1)} \neq m^{(2)}$ and there are no flags, perform one syndrome measurement round using non-flag circuits for error correction and apply $E_\text{min}(s^{(1)})$. Repeat the operator measurement using a non-flag circuit, then output $m^{(3)}$. 
	
	\item If $f_0$ does not flag but $f_i$ flags (with $i \ge 1$) during round one, stop. Repeat the operator measurement using a non-flag circuit then output $m^{(2)}$. If there are no flags during round one but $f_i$ flags and $f_0$ does not flag during round two, output $m^{(1)}$.
	
	\item If $f_0$ flags at round $r_1$ anytime during the protocol, stop and do one of the followings:
	\begin{enumerate}
		\item If $f_i$ does not flag for all $i \geq 1$, repeat the operator measurement using a non-flag circuit and output $m^{(r_1+1)}$.
		
		\item If there is only one $i$ such that $f_i$ flags (with $i \geq 1$), apply $I^{\otimes c} \otimes P_{i+1}^{\otimes a_{i+1}} \otimes I^{\otimes b_{i+1}} \otimes \cdots \otimes P_m^{\otimes a_m} \otimes I^{\otimes b_m}$ to the data, where $c = \sum_{j=1}^i a_j+b_j$. Perform the syndrome measurement using non-flag circuits for error correction yielding syndrome $s^{(r_2)} = (s^{(r_2)}_x|s^{(r_2)}_z)$. 
		\begin{enumerate}
			\item If $P_i=Z$, apply $E_z \in \mathcal{E}^z_{l,n}$ that satisfies $s(E_z) = s^{(r_2)}_x$  where $l=n-c+b_i$, followed by $E_\text{min}(s^{(r_2)}_z)$.
			\item If $P_i=X$, apply $E_x \in \mathcal{E}^x_{l,n}$ that satisfies $s(E_x) = s^{(r_2)}_z$  where $l=n-c+b_i$, followed by $E_\text{min}(s^{(r_2)}_x)$.
			\item If $P_i=Y$, apply $E \in \mathcal{E}^P_{l,n}$ that satisfies $s(E) = s^{(r_2)}$ where $l=n-c+b_i$.
		\end{enumerate}
		Afterwards, repeat the operator measurement using a non-flag circuit, then output $m^{(r_1+1)}$.
		
		\item If there is an $i$ such that $f_i$ and $f_{i+1}$ flag, perform the same sequence of operations as in 4(b).
	\end{enumerate}
	
\end{enumerate}

To see that \cref{Def:FTmeasDef} is satisfied, we will assume that the weight of an input error $v_1$ and the number of faults during the protocol  $v_2$ satisfy $v_1+v_2 \leq 1$. Similar to the FTEC protocol, $f_0$ and at least one $f_i$ (with $i\geq 1$) will flag whenever a fault in any CNOT gate causes $Z$ error on $m_0$. If there is no flags, a single fault can introduce error of weight at most one. If the measurement result is repeated twice, then there is no fault in the circuits. However, the measurement result might be incorrect due to the input error. By performing full syndrome measurement twice with flag circuits, we can determine from $s^{(1)}$ and $s^{(2)}$ whether these is no input error, there is a fault during syndrome measurement, or there is an input error of weight 1. The procedure in 1(a), 1(b), and 1(c) can correct possible errors and output the right operator measurement result with corresponding codeword after projective measurement.

Now let us consider the case that there is no flags but $m^{(1)} \neq m^{(2)}$. This is the case where a fault occurred in either the first or second round. Therefore, performing error correction and repeating the operator measurement can give the correct result.

Next, consider the case that there is only one flag, either $f_0$ or $f_i$ with $i\geq 1$. The fault could be a measurement error, idle qubit error on the ancilla $f_0$ or $f_i$, or an error on the control qubit of the CNOT gate interacting with $f_0$ or $f_i$. Repeating the operator measurement can give the right result. Note that if $f_i$ flags during round two, then the result obtained during round one corresponds to the right outcome.

Now let us consider the case where a fault happens on a red CNOT introducing a $Z$ error on the ancilla $m_0$ and a $\tilde{P}$ error on the data qubit where $\tilde{P} \in \{I,X,Y,Z\}$. If the fault occurs on the $i^\text{th}$ sub-block, then $f_0$ and only one $f_i$ with $i \ge 1$ will flag. Applying $I^{\otimes c} \otimes P_{i+1}^{\otimes a_{i+1}} \otimes I^{\otimes b_{i+1}} \otimes \cdots \otimes P_m^{\otimes a_m} \otimes I^{\otimes b_m}$ to the data guarantees that the resulting error is in the form $I^{\otimes c-a_i-b_i} \otimes I^{\otimes c-1}\otimes \tilde{P} \otimes P_i^{\otimes a_i-c} \otimes I^{\otimes b_i} \otimes I^{\otimes n-c}$ (where $c = \sum_{j=1}^i a_j+b_j$). If $P_i=Z$ (or $P_i=X$), the resulting error is a product of consecutive error in $\mathcal{E}^z_{l,n}$ (or $\mathcal{E}^x_{l,n}$) and $X$-type error (or $Z$-type error) of weight one, where $l=n-c+b_i$. If $P_i=Y$, the resulting error is a consecutive error in $\mathcal{E}^P_{l,n}$. By \cref{theo:MainTheorem}, errors in  $\mathcal{E}^x_{l,n}$, $\mathcal{E}^z_{l,n}$, and $\mathcal{E}^P_{l,n}$ are distinguishable. Therefore, performing a full syndrome measurement followed by appropriate error correction as in 4(b) will remove the error, and repeating the operator measurement gives the correct outcome. The case that a fault occurs on a blue or green CNOT corresponds to either 4(b) or 4(c), and the same correction procedure can be applied.

A list of possible faults during the flag-FT operator measurement protocol and corresponding correction procedures is given in \cref{tab:FTM_fault} in \cref{app:FaultTable}.

The flag-FT measurement protocol described above is for a measurement of an operator commuting with all generators which acts in one code block. Surprisingly, the protocol also works for an operator acting on two or more code blocks. The measurement of such operator can be done by treating parts of the operator acting on different code blocks as operators from different sub-blocks. For example, let $C_p$ and $C_q$ be cyclic CSS codes of distance 3 satisfying \cref{theo:MainTheorem}, and let $P$ and $Q$ be Pauli operators acting on $C_p$ and $C_q$, respectively. The measurement of $P\otimes Q$ on the code $C_p\otimes C_q$ can be done by using a 1-flag circuit given in \cref{fig:GeneralWflagFig}, where $P$ and $Q$ are treated as operators from the different sub-blocks. Observe that if $f_0$ flags and at least one $f_i$ flags (the 4(b) or 4(c) case), the resulting error after appropriate operation will become a consecutive error on either first or second code blocks. Since $C_p$ and $C_q$ are both cyclic, we can determine the error by performing subsequent syndrome measurement on only $C_p$ or $C_q$, depending on the sub-block in which the fault occurs. After that, the correct measurement result can be obtained by a subsequent operator measurement.

\section{Applications of fault-tolerant operator measurement protocol}
\label{sec:Applications}%
A measurement of an operator commuting with all generators can be used as a subroutine in numerous quantum information processing techniques such as state preparation and quantum gate implementation. Since the fault-tolerant measurement protocol described in \cref{sec:MeasProtocol} is applicable on two or more code blocks, information processing between code blocks is possible. In this section we briefly describe some important techniques which make fault-tolerant computation on cyclic CSS codes possible, including logical EPR state preparation, teleportation, and quantum computation on logical qubits.

Let us consider an EPR state $\frac{|00\rangle+|11\rangle}{\sqrt{2}}$. This is a +1 eigenstate of operators $X \otimes X$ and $Z \otimes Z$. Let $C_p$ and $C_q$ be \codepar{n_1,k_1,d_1} and \codepar{n_2,k_2,d_2} cyclic CSS codes satisfying \cref{theo:MainTheorem} with stabilizer generating sets $\{g_{i_1}^p\}$ and $\{g_{i_2}^q\}$, respectively. Suppose that we want to prepare a state
\begin{align}
\frac{|\bar{0}\rangle_{p,i}|\bar{0}\rangle_{q,j}+|\bar{1}\rangle_{p,i}|\bar{1}\rangle_{q,j}}{\sqrt{2}}, \label{eqn:EPR}%
\end{align}
which is an EPR state between the $i^\text{th}$ logical qubit of $C_p$ and the $j^\text{th}$ logical qubit of $C_q$. This can be done by performing projective measurements with respect to stabilizer generators $\{g_{i_1}^p\otimes I, I \otimes g_{i_2}^q\}$ and logical operators $\bar{X}_{p,i} \otimes \bar{X}_{q,j}$ and $\bar{Z}_{p,i} \otimes \bar{Z}_{q,j}$ on a totally mixed state, where $\bar{X}_{p,i}$ and $\bar{X}_{q,j}$ (or $\bar{Z}_{p,i}$ and $\bar{Z}_{q,j}$) are logical $X$ (or logical $Z$) operators on $i^\text{th}$ logical qubit of $C_p$ and $j^\text{th}$ logical qubit of $C_q$, respectively. Since the measurement protocol described in \cref{sec:MeasProtocol} is a fault-tolerant protocol, the state in \cref{eqn:EPR} can be prepared fault-tolerantly.

In conventional quantum teleportation, an EPR state and Bell measurement are required. Here we will examine a process to fault-tolerantly perform quantum teleportation of logical data between two code blocks. The scheme for logical qubit teleportation is shown in \cref{fig:qubitTP}. Suppose that we would like to teleport the $i^\text{th}$ logical qubit of $C_p$ to the $j^\text{th}$ logical qubit of $C_q$, first an EPR state $\frac{|\bar{0}\rangle_{q,j}|\bar{0}\rangle_{q,j}+|\bar{1}\rangle_{q,j}|\bar{1}\rangle_{q,j}}{\sqrt{2}}$ prepared on $C_q\otimes C_q$ is required. The logical qubit teleportation can be done by performing a Bell measurement with respect to $\bar{X}_{p,i}\otimes\bar{X}_{q,j}$ and $\bar{Z}_{p,i}\otimes\bar{Z}_{q,j}$ between $C_p$ and the first block of $C_q$. The teleported logical qubit can be obtained in the second block of $C_q$ by operating an appropriate logical Pauli operator $\bar{P}_{q,j}$ depending on the Bell measurement result. Note that the Bell measurement can be done fault-tolerantly using the measurement protocol in \cref{sec:MeasProtocol}, and logical Pauli operators are transversal (therefore, fault tolerant). Thus, fault-tolerant teleportation between two code blocks can be achieved.

Now let us consider fault-tolerant computation on cyclic CSS codes. It is known that for any error correcting code, by the Eastin-Knill theorem \cite{EK09}, at least one logical gate in a universal gate set cannot be implemented transversely. For such gates, other fault-tolerant techniques must be performed, which can require a significant amount of resources. Fortunately, fault-tolerant implementations of logical Clifford gates on distance-3 cyclic CSS codes can be achieved via quantum gate teleportation (see \cite{GC99} for the details of quantum gate teleportation). For example, suppose that we would like to perform a logical Hadamard gate $\bar{H}_i$ on the code $C_p$. This can be achieved by preparing a codeword which is an eigenstate of $\bar{X}_i\otimes\bar{Z}_i$ and $\bar{Z}_i\otimes\bar{X}_i$ on $C_p \otimes C_p$, performing logical qubit teleportation, and operating a logical Pauli operator $\bar{H}_{p,i}\bar{P}_{p,i}\bar{H}_{p,i}$ depending on the result from the Bell measurement in qubit teleportation as illustrated in \cref{fig:gateTP}. Also, logical $R_{\pi/4} = \text{diag(1,i)}$ and logical CNOT gates on any logical qubits can be performed in a similar way. The Clifford group can be generated by $\{H,R_{\pi/4},\text{CNOT}\}$ \cite{CRSS97,Gottesman98b}. Thus, our scheme is applicable to any Clifford operation. It is known that universal quantum computation can be achieved by Clifford gates and any other gate not in the Clifford group \cite{NRS01}. However, performing logical non-Clifford gates will require different techniques such as the ones presented in \cite{CC18}.

\begin{figure}[htbp]
	\centering
	\begin{subfigure}{0.45\textwidth}
		\includegraphics[width=\textwidth]{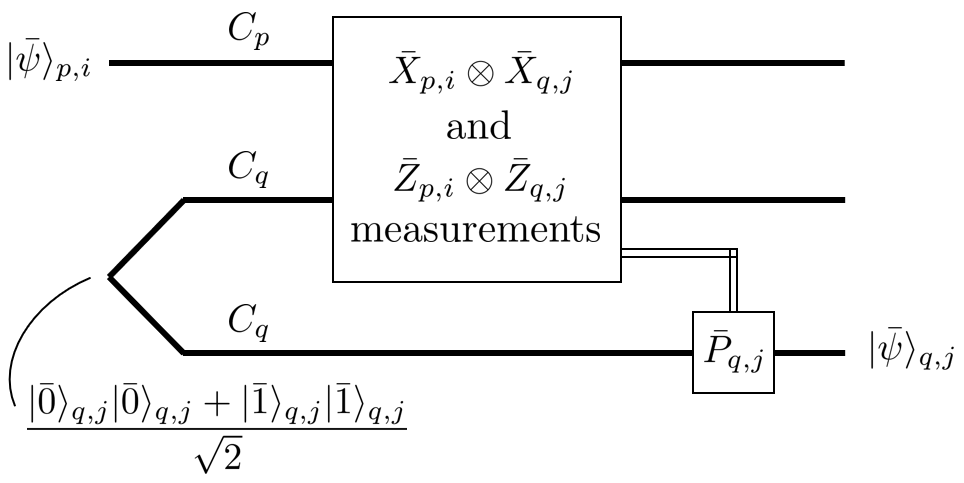}
		\caption{}
		\label{fig:qubitTP}%
	\end{subfigure}	
	\begin{subfigure}{0.49\textwidth}
		\includegraphics[width=\textwidth]{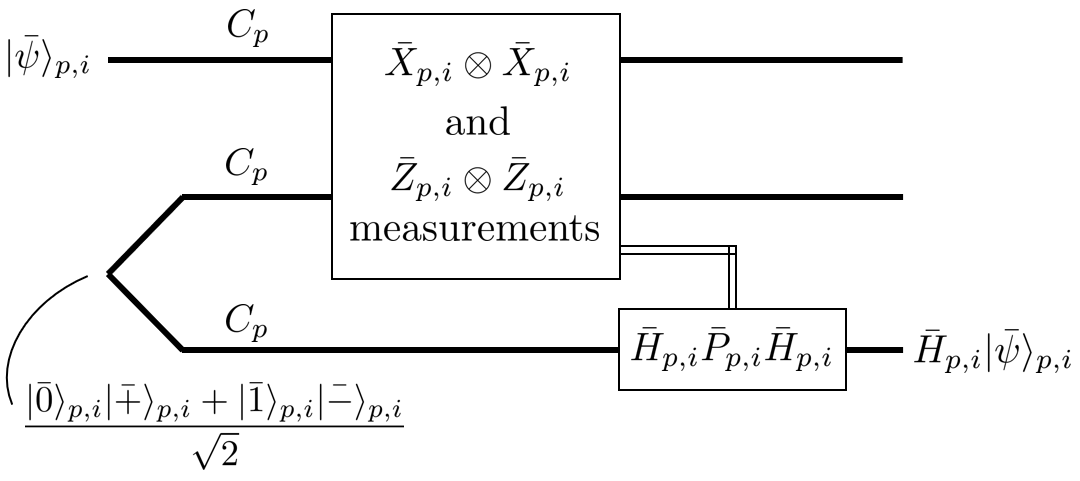}
		\caption{}
		\label{fig:gateTP}%
	\end{subfigure}
	\caption{Schemes for teleportation and Clifford gate implementation on cyclic CSS codes. A bold line represents a block of code, while a double line represents classical information. In \cref{fig:qubitTP}, the $i^\text{th}$ logical qubit of $C_p$ is teleported to the $j^\text{th}$ logical qubit of $C_q$. In \cref{fig:gateTP}, a logical Hadamard gate $\bar{H}_i$ is performed on the $i^\text{th}$ logical qubit of $C_p$ via quantum gate teleportation.}
	\label{fig:TPScheme}%
\end{figure}

\section{Examples of cyclic CSS codes}
\label{sec:ExCyclicCSS}%

In this section, some examples of cyclic CSS codes satisfying \cref{theo:MainTheorem} are given. A first example is the \codepar{7,1,3} quantum Hamming code. This code is constructed from a classical [7,4,3] Hamming code (with $C_x=C_z$). A check polynomial of the [7,4,3] Hamming code in cyclic form is
\begin{align}
h(x) = 1+x^2+x^3+x^4.
\end{align}

In fact, any classical Hamming code can be made cyclic \cite{MS77}. Thus, any CSS code constructed from a classical $[2^r-1,2^r-1-r,3]$ Hamming code with $C_x=C_z$ satisfies \cref{theo:MainTheorem}, and can be used in the flag-FTEC protocol and the flag-FT measurement protocol described in this work. 

Another example of cyclic CSS codes satisfying \cref{theo:MainTheorem} is the \codepar{30,14,3} code constructed from a classical [30,22,3] cyclic code with a check polynomial 
\begin{align}
h(x) = 1+x^2+x^4+x^6+x^{10}+x^{14}+x^{16}+x^{22}.
\end{align}
\begin{table*}[htbp]
	\begin{center}
		\begin{tabular}{| c | c || c | c |}
			\hline
			$\bar{X}_1$\rule{0pt}{2.5ex}  & $X_{1}X_{11}X_{21}$ & $\bar{Z}_1$ & $Z_{1}Z_{11}Z_{21}$\\
			\hline
			$\bar{X}_2$\rule{0pt}{2.5ex}  & $X_{2}X_{12}X_{22}$ & $\bar{Z}_2$ & $Z_{2}Z_{12}Z_{22}$\\
			\hline
			$\bar{X}_3$\rule{0pt}{2.5ex}  & $X_{3}X_{13}X_{23}$ & $\bar{Z}_3$ & $Z_{3}Z_{13}Z_{23}$\\
			\hline
			$\bar{X}_4$\rule{0pt}{2.5ex}  & $X_{4}X_{14}X_{24}$ & $\bar{Z}_4$ & $Z_{4}Z_{14}Z_{24}$\\
			\hline
			$\bar{X}_5$\rule{0pt}{2.5ex}  & $X_{5}X_{15}X_{25}$ & $\bar{Z}_5$ & $Z_{5}Z_{15}Z_{25}$\\
			\hline
			$\bar{X}_6$\rule{0pt}{2.5ex}  & $X_{6}X_{16}X_{26}$ & $\bar{Z}_6$ & $Z_{6}Z_{16}Z_{26}$\\
			\hline
			$\bar{X}_7$\rule{0pt}{2.5ex}  & $X_{7}X_{17}X_{27}$ & $\bar{Z}_7$ & $Z_{7}Z_{17}Z_{27}$\\
			\hline
			$\bar{X}_8$\rule{0pt}{2.5ex}  & $X_{8}X_{18}X_{28}$ & $\bar{Z}_8$ & $Z_{8}Z_{18}Z_{28}$\\
			\hline
			$\bar{X}_9$\rule{0pt}{2.5ex}  & $X_{9}X_{19}X_{29}$ & $\bar{Z}_9$ & $Z_{9}Z_{19}Z_{29}$\\
			\hline
			$\bar{X}_{10}$\rule{0pt}{2.5ex}  & $X_{10}X_{20}X_{30}$ & $\bar{Z}_{10}$ & $Z_{10}Z_{20}Z_{30}$\\
			\hline
			$\bar{X}_{11}$\rule{0pt}{2.5ex}  & $X_{1}X_{7}X_{9}X_{11}X_{17}X_{19}$ & $\bar{Z}_{11}$ & $Z_{11}Z_{17}Z_{19}Z_{21}Z_{27}Z_{29}$\\
			\hline
			$\bar{X}_{12}$\rule{0pt}{2.5ex}  & $X_{2}X_{8}X_{10}X_{12}X_{18}X_{20}$ & $\bar{Z}_{12}$ & $Z_{12}Z_{18}Z_{20}Z_{22}Z_{28}Z_{30}$\\
			\hline
			$\bar{X}_{13}$\rule{0pt}{2.5ex}  & $X_{11}X_{17}X_{19}X_{21}X_{27}X_{29}$ & $\bar{Z}_{13}$ & $Z_{1}Z_{7}Z_{9}Z_{11}Z_{17}Z_{19}$\\
			\hline
			$\bar{X}_{14}$\rule{0pt}{2.5ex}  & $X_{12}X_{18}X_{20}X_{22}X_{28}X_{30}$ & $\bar{Z}_{14}$ & $Z_{2}Z_{8}Z_{10}Z_{12}Z_{18}Z_{20}$\\
			\hline
		\end{tabular}
		\caption{A choice of logical operators for the \codepar{30,14,3} code.}
		\label{tab:logicalop}%
	\end{center}
\end{table*}
The [30,22,3] code and other classical codes satisfying $C^\perp \subseteq C$ are given in Table 1 of \cite{QuantumCodesCyclic}. (A method of finding the check polynomial of a classical cyclic code is discussed in \cite{MS77}.) One possible choice of logical operators for the \codepar{30,14,3} code is given in \cref{tab:logicalop}. The advantages of the \codepar{30,14,3} code are that its encoding rate is high ($k/n = 14/30$), and the logical operators of the first ten logical qubits have a simple form, which make them easily accessible. 
 
\section{Discussion and conclusion}
\label{sec:Discussion}%
In this work we used the symmetries of CSS codes constructed from classical cyclic codes to prove that errors written in consecutive form (as in \cref{def:Consecutive}) can be distinguished. From these properties, we can obtain a 1-flag circuit along with a flag-FTEC protocol which satisfies the fault-tolerance criteria of \cref{Def:FaultTolerantDef} when there is at most one fault. The 1-flag circuit requires only four ancilla qubits. This number does not grow as the block length gets larger, making our protocol advantageous in the implementation where resources are limited. We note that not all cyclic CSS codes are Hamming codes and therefore the methods in \cite{CR17v1} (which apply to perfect codes) cannot be directly applied, thus providing further motivation for our work. 

In general, cyclic CSS codes do not satisfy the sufficient condition required for flag fault-tolerance presented in \cite{CB17} (one example is the family of Hamming codes which can be made cyclic). Nevertheless, using the techniques presented in this paper, a flag-FTEC protocol can still be achieved.  

Furthermore, we have shown how logical Pauli operators of cyclic CSS codes can be fault-tolerantly measured using the flag techniques discussed in \cref{sec:MeasProtocol}. The flag-FT operator measurement protocol satisfies the fault-tolerance criteria of \cref{Def:FTmeasDef} when there is at most one fault. We then showed in \cref{sec:Applications} how one can fault-tolerantly perform quantum gate teleportation to implement logical Clifford operators on any given logical qubit for codes which encode multiple logical qubits. Examples of cyclic CSS codes with large encoding rates are provided in \cref{sec:ExCyclicCSS}.

Note that for all CSS codes, the stabilizer generators being measured are of the form $I^{\otimes n-m} \otimes X^{\otimes m}$ or $I^{\otimes n-m} \otimes Z^{\otimes m}$ up to qubit permutations. Thus data qubit errors arising from faulty CNOT gates will be expressed in consecutive form. The errors of this form are distinguishable iff the sub-matrices of the $X$ and $Z$ stabilizers  satisfy \cref{Lemma:Lem1}. In our work, we use the symmetry of the cyclic codes to simplify \cref{Lemma:Lem1} into \cref{Lemma:Lem3}. We believe that \cref{Lemma:Lem1} can be simplified by using symmetries found in other families of quantum codes. With appropriate $t$-flag circuits and operations depending on the flag measurement outcome, this may lead to new flag fault-tolerant protocols.

Another interesting avenue is finding non-cyclic quantum codes for which a version of \cref{theo:MainTheorem} can be applied. We note that for such codes, the same 1-flag circuit as in \cref{fig:GeneralWflagFig} along with the flag-FTEC protocol of \cref{sec:Protocol} and the flag-FT measurement protocol of \cref{sec:MeasProtocol} can be used. The reason is that the key property used by these schemes is based on the distinguishability of consecutive errors.

Note that there are quantum cyclic codes which are not CSS codes for which flag fault-tolerant schemes are still possible. For instance, a flag-FTEC protocol for the \codepar{5,1,3} code was devised in \cite{CR17v1}. We believe that it could be interesting to generalize the ideas presented in this work to non-CSS cyclic quantum codes. However, we leave this problem for future work. 

The flag fault-tolerant protocols for cyclic CSS codes presented in this work are based on the assumption that the qubit measurement and state preparation must be fast since we reuse some flag qubits in the protocols (as we can see in \cref{fig:GeneralWflagFig}). If we do not reuse flag qubits, however, the number of required ancillas will be $m+2$ for an operator being measured of the form $P = P_1^{\otimes a_1} \otimes I^{\otimes b_1} \otimes P_2^{\otimes a_2} \otimes I^{\otimes b_2} \otimes \cdots \otimes P_m^{\otimes a_m} \otimes I^{\otimes b_m}$ instead of $4$.

One important feature of flag fault-tolerant protocols is that the number of required ancillas is very small compared to other fault-tolerant schemes. We believe that if fewer ancillas are required, the accuracy threshold will increase since the number of locations will decrease in total. However, we should point out that subsequent syndrome measurements are also required in a flag fault-tolerant protocol and may increase the total number of locations in the protocol. The answer of whether the accuracy threshold for a flag fault-tolerant protocol is greater or smaller compared to other fault-tolerant schemes is still unknown.

Lastly, we point out that cyclic CSS codes which satisfy the condition in \cref{theo:MainTheorem} are not limited to distance-3 codes. Therefore, interesting future work would be to use the methods of \cite{CB17} to obtain flag fault-tolerant schemes for higher-distance codes. In particular, the main challenge stems from finding $t$-flag circuits as in \cref{fig:GeneralWflagFig} for $t > 1$.

\section{Acknowledgements}
T.T. acknowledges the support of The Queen Sirikit Scholarship under The Royal Patronage of Her Majesty Queen Sirikit of Thailand. C.C. acknowledges the support of NSERC through the PGS D scholarship. D.L. is supported by an NSERC Discovery grant and a CIFAR research grant via the Quantum Information Science program.

\newpage 

\bibliographystyle{ieeetr}
\bibliography{bibtex_chamberland}

\clearpage
\appendix

\section{Proof of the lemmas}
\label{app:LemmaProofs}%

\textbf{Proof of \cref{Lemma:Lem1}}: We will prove that $\mathcal{E}^z_{0,n}$ is distinguishable by $C$ iff for all $p,q\in\{0,1,...,n-1\}$ such that $p>q$, there exists $i\in\{1,\dots,r_x\}$ such that $x_{i,n-p+1}\oplus\dots\oplus x_{i,n-q}=1$. Consider errors $E_{p}=I^{\otimes n-{p}} \otimes Z^{\otimes{p}} $ and $E_{q}=I^{\otimes n-{q}} \otimes Z^{\otimes{q}}$ where $p,q\in\{0,1,\dots,n-1\}, p>q$. Let $s(E_{p}),s(E_{q})\in\mathbb{Z}_2^r$ be error syndromes corresponding to errors $E_{p}$ and $E_{q}$, respectively. By \cref{Def:DisErr}, $E_{p}$ and $E_{q}$ are distinguishable by $C$ iff $s(E_{p})\neq s(E_{q})$, i.e., there exists $i \in \{1,2, \cdots, r_x \}$ such that $s(E_{p})_i\neq s(E_{q})_i$ (here $i$ corresponds to the $i^{\text{th}}$ component of $s(E_{p})$ and $S(E_{q})$). From the parity check matrix $H_x$, the $i^\text{th}$ component of $s(E_{p})$ and $s(E_{q})$ is given by 
\begin{align}
s(E_{p})_i &= x_{i,n-p+1}\oplus x_{i,n-p+2}\oplus\cdots\oplus x_{i,n}, 
\label{eq:App1} \\
s(E_{q})_i &= x_{i,n-q+1}\oplus x_{i,n-q+2}\oplus\cdots\oplus x_{i,n}.
\label{eq:App2}%
\end{align}

From \cref{eq:App1,eq:App2}, we have that
\begin{align}
s(E_{p})_i\neq s(E_{q})_i &\Leftrightarrow s(E_{p})_i\oplus s(E_{q})_i = 1 \nonumber\\
&\Leftrightarrow x_{i,n-p+1}\oplus\cdots\oplus x_{i,n-q}=1.
\end{align}
Thus, $\mathcal{E}^z_{0,n}$ is distinguishable by $C$ iff for all $p,q\in\{0,1,...,n-1\}$ such that $p>q$, there exists $i \in \{ 1,2, \cdots, r_x \}$ such that
\begin{align}
x_{i,n-p+1}\oplus\dots\oplus x_{i,n-q}=1.
\end{align}

The proof of the statement for $\mathcal{E}^x_{0,n}$ is similar.

Now we will proof that $\mathcal{E}^P_{0,n}$ is distinguishable by $C$ iff both $\mathcal{E}^z_{0,n}$ and $\mathcal{E}^x_{0,n}$ are distinguishable by $C$. Let $X_p = I^{\otimes n-p} \otimes X^{\otimes p}$ and $Z_q = I^{\otimes n-q} \otimes Z^{\otimes q}$, where $p,q \in \{0,\dots,n-1\}$. Observe that any element of $\mathcal{E}^P_{0,n}$ is of the form $E_{p,q} = X_p \cdot Z_q$ where $X_p \in \mathcal{E}^x_{0,n}$ and $Z_q \in \mathcal{E}^z_{0,n}$. The syndrome of $E_{p,q}$ is $s(E_{p,q}) = (s(X_p)|s(Z_q))$. If $\mathcal{E}^P_{0,n}$ is distinguishable by $C$, i.e., $s(E_{p_1,q_1}) \neq s(E_{p_2,q_2})$ for all choices of $p_1,p_2,q_1,q_2$ such that $(p_1,q_1) \neq (p_2,q_2)$, then we have that any pair of $X_{p_1}$ and $X_{p_2}$ and any pair of $Z_{q_1}$ and $Z_{q_2}$ are distinguishable. Conversely, if any pair of $X_{p_1}$ and $X_{p_2}$ and any pair of $Z_{q_1}$ and $Z_{q_2}$ are distinguishable, then any pair of $E_{p_1,q_1}$ and $E_{p_2,q_2}$ will have different syndromes. This implies Statement 3.\\

\textbf{Proof of \cref{Lemma:Cyclic_gen}}: Suppose that the stabilizer group of $C$ can be generated by $\{g_1,g_2,\dots,g_{n-k}\}$. Since $C$ is a CSS code, we will first assume that the generators $g_i$'s are either $X$-type or $Z$-type, denoted as $g^x_i$ or $g^z_i$. Let $H_x$ and $H_z$ be $X$ and $Z$ stabilizer matrices of the code $C$ in symplectic representation, and let $C^\perp_x$ and $C^\perp_z$ be the classical codes generated by $H_x$ and $H_z$, respectively. Observe that any element of $C$ in symplectic representation is of the form $(x|z)$ where $x \in C^\perp_x$ and $z \in C^\perp_z$. For any choice of $l \in \{0,1,\dots,n-1\}$, let $\tilde{H}_x$ (or $\tilde{H}_z$) be the parity check matrix corresponding to $\mathcal{L}(g^x_{i},l)$'s (or $\mathcal{L}(g^z_{i},l)$'s). We find that the code $\tilde{C}_x^\perp$ generated by $\tilde{H}_x$ (or $\tilde{C}_z^\perp$ generated by $\tilde{H}_z$) differs from $C_x^\perp$ (or $C_z^\perp$) by an $l$-step left cyclic permutation. However, since $C_x^\perp$ and $C_z^\perp$ are cyclic codes, we have that $\tilde{C}_x^\perp = C_x^\perp$ and $\tilde{C}_z^\perp = C_z^\perp$. Therefore, $\{\mathcal{L}(g^x_{1},l),\dots,\mathcal{L}(g^x_{r_x},l),\mathcal{L}(g^z_{1},l),\dots,\mathcal{L}(g^z_{r_z},l)\}$ and $\{g^x_{1},\dots,g^x_{r_x},g^z_{1},\dots,g^z_{r_z}\}$ generate the same stabilizer group for any $l \in \{0,1,\dots,n-1\}$. 

In general, some generators of the stabilizer group of $C$ might be neither $X$-type nor $Z$-type. The following transformations of the generators preserve the stabilizer group, and the last set of generators is the cyclic shifts of the original: (1) Transform the given generators to either $X$-type or $Z$-type. This corresponds to appropriate reversible row operations on the binary symplectic representation of $C$ to obtain the block diagonal form,
\begin{equation}
\begin{pmatrix}
A &|& 0 \\
0 &|& B 
\end{pmatrix}.
\end{equation}
(2) Cyclic shifts of these new generators also generate the same stabilizer group. (3) Reversing the transformation in step (1) (now applied to the generators after step (2)) preserves the stabilizer group. The resulting generators are cyclic shifts of the original.

\textbf{Proof of \cref{Lemma:Lem3}}: First we will prove that $\mathcal{E}^z_{0,n}$ is distinguishable by $C$ iff for all $u_x\in\{2,3,\dots,n\}$, there exists $i\in\{1,\dots,r_x\}$ such that $x_{i,u_x}\oplus\dots\oplus x_{i,n}=1$.
Applying \cref{Lemma:Lem1}, we would like to prove that for all $p,q\in\{0,...,n-1\}$ such that $p>q$, there exists $i \in \{ 1, \cdots, r_x \}$ such that $x_{i,n-p+1}\oplus\dots\oplus x_{i,n-q}=1$  iff for all $u_x\in\{2,3,\dots,n\}$, there exist $i' \in \{1, \cdots , r_x \}$ such that $x_{i',u_x}\oplus\dots\oplus x_{i',n}=1$.

\noindent $(\Rightarrow)$ By choosing $q=0$ and $p=n-u+1$, the proof is trivial.

\noindent $(\Leftarrow)$ Assume by contradiction that there exists a pair of $p,q \in \{0,1,\dots,n{-}1\}$ with $p>q$ such that $x_{i,n-p+1}\oplus\dots\oplus x_{i,n-q}=0$ for all $i$; i.e., there exists a pair of errors $E_p=I^{\otimes n-p}\otimes Z^{\otimes p}$ and $E_q=I^{\otimes n-q}\otimes Z^{\otimes q}$ which cannot be distinguished by any generator of $C$. Let $C$ be generated by $\{g_1,\dots,g_r\}$. By \cref{Lemma:Cyclic_gen}, we can construct a new generator set $\{\tilde{g}_1,\dots,\tilde{g}_r\}$ of $C$ where $\tilde{g}_i = \mathcal{L}(g_i,q)$ for all $i$. Let the $X$-part of $\sigma(\tilde{g}_i)$ be  $(\tilde{x}_{i,1},\dots,\tilde{x}_{i,n})=(x_{i,q+1},\dots,x_{i,n},x_{i,1},\dots,x_{i,q})$. Note that $\tilde{x}_{n-p+1}=x_{i,n-(p-q)+1}$ and $\tilde{x}_{n-q}=x_{i,n}$. The assumption implies that $E_p$ and $E_q$ cannot be distinguished by any $\tilde{g}_i$ as well. This gives
\begin{align}
\tilde{x}_{i,n-p+1}\oplus\dots\oplus \tilde{x}_{i,n-q}&=0, 
\end{align}
or equivalently,
\begin{align}
x_{i,n-(p-q)+1}\oplus\dots\oplus x_{i,n}&=0.
\end{align}
Let $u_x = n-(p-q)+1$. Therefore, there exists $u_x \in \{2,3,\dots,n\}$ such that $x_{i,u_x}\oplus\dots\oplus x_{i,n}=0$ for all $i$.

The proof of statement for $\mathcal{E}^x_{0,n}$ is similar to the proof of statement for $\mathcal{E}^z_{0,n}$, while the proof of statement for $\mathcal{E}^P_{0,n}$ is similar to the proof of Statement 3 in \cref{Lemma:Lem1}.

\begin{table*}[htbp]
	\begin{center}
		\begin{tabular}{| c | c |}
			\hline
			Type of faults & Correction procedure\\
			\hhline{|=|=|}
			No fault & 1 \\
			\hline
			Qubit or measurement fault on $m_0$ & 1 or 2 \\
			\hline
			Qubit or measurement fault on $f_0$	& 1 or 4(a) \\
			\hline
			Qubit or measurement fault on $f_i$	& 1 or 3 \\
			\hline
			Red CNOT fault with $I$ or $X$ error on the target qubit & 2 \\
			\hline
			Red CNOT fault with $Y$ or $Z$ error on the target qubit & 4(b) \\
			\hline
			Blue or green CNOT fault with $I$ or $X$ error on the target qubit & 1 or 2 or 3 \\
			\hline
			Blue or green CNOT fault with $Y$ or $Z$ error on the target qubit & 4(b) or 4(c) \\
			\hline
			Orange CNOT fault with $I$ or $X$ error on the target qubit	& 1 or 2 or 4(a) \\
			\hline
			Orange CNOT fault with $Y$ or $Z$ error on the target qubit	& 2 or 4(a) \\		
			\hline
		\end{tabular}
		\caption{Possible faults during the flag-FTEC protocol in \cref{sec:Protocol} and their corresponding correction procedures. Here we assume that the number of faults $v_2$ is at most 1.}
		\label{tab:FTEC_fault}%
	\end{center}
\end{table*}
\begin{table*}[htbp]
	\begin{center}
		\begin{tabular}{| c | c |}
			\hline
			Type of faults & Correction procedure\\
			\hhline{|=|=|}
			\multicolumn{2}{|c|}{------ No fault during operator measurement ------} \\
			\hline
			No input error, no fault during syndrome measurement & 1(a) \\
			\hline
			No input error, one fault during syndrome measurement & 1(b) or 1(c) \\
			\hline
			weight-1 input error, no fault syndrome measurement & 1(c) \\
			\hline
			\multicolumn{2}{|c|}{------ One fault during operator measurement ------} \\
			\hline
			Qubit or measurement fault on $m_0$ & 2 \\
			\hline
			Qubit or measurement fault on $f_0$	& 4(a) \\
			\hline
			Qubit or measurement fault on $f_i$	& 3 \\
			\hline
			Red CNOT fault with $I$ or $X$ error on the target qubit & 2 \\
			\hline
			Red CNOT fault with $Y$ or $Z$ error on the target qubit & 4(b) \\
			\hline
			Blue or green CNOT fault with $I$ or $X$ error on the target qubit & 1(a) or 2 or 3 \\
			\hline
			Blue or green CNOT fault with $Y$ or $Z$ error on the target qubit & 4(b) or 4(c) \\
			\hline
			Orange CNOT fault with $I$ or $X$ error on the target qubit	& 1(a) or 2 or 4(a) \\
			\hline
			Orange CNOT fault with $Y$ or $Z$ error on the target qubit	& 2 or 4(a) \\		
			\hline
		\end{tabular}
		\caption{Possible faults during the flag-FT operator measurement protocol in \cref{sec:MeasProtocol} and their corresponding correction procedures. Here we assume that the number of input errors $v_1$ and the number of faults $v_2$ satisfy $v_1+v_2 \leq 1$.}
		\label{tab:FTM_fault}%
	\end{center}
\end{table*}

We already proved statements for $\mathcal{E}^z_{0,n}$, $\mathcal{E}^x_{0,n}$, and $\mathcal{E}^P_{0,n}$. We will generalize the statements to $\mathcal{E}^z_{l,n}$, $\mathcal{E}^x_{l,n}$, and $\mathcal{E}^P_{l,n}$ for any $l \in \{0,\dots,n-1\}$. Let $\tilde{C}$ by a cyclic CSS code generated by $\{\mathcal{L}(g^x_{1},l),\dots,\mathcal{L}(g^x_{r_x},l),\mathcal{L}(g^z_{1},l),\dots,\mathcal{L}(g^z_{r_z},l)\}$. Observe that by qubit reordering, $\mathcal{E}^P_{l,n}$ is distinguishable by $\tilde{C}$ iff $\mathcal{E}^P_{0,n}$ is distinguishable by $C$. Since $\tilde{C}$ and $C$ are the same code by \cref{Lemma:Cyclic_gen}, we have that $\mathcal{E}^P_{l,n}$ is distinguishable by $C$ for any $l \in \{0,\dots,n-1\}$ iff $\mathcal{E}^P_{0,n}$ is distinguishable by $C$. The proof is also applied to $\mathcal{E}^z_{l,n}$ and $\mathcal{E}^x_{l,n}$.

\section{Lists of possible faults during FTEC and FT operator measurement protocols}
\label{app:FaultTable}%
The lists of possible faults during the flag-FTEC protocol (see \cref{sec:Protocol}) and the flag-FT operator measurement protocol (see \cref{sec:MeasProtocol}) are given in \cref{tab:FTEC_fault,tab:FTM_fault}, respectively. The corresponding correction procedure for each type of faults refers to the step of either flag-FTEC protocol or flag-FT operator measurement protocol. Here we assume that $v_2 \leq 1$ in the flag-FTEC protocol and $v_1+v_2 \leq 1$ in the flag-FT operator measurement protocol, where $v_1$ is the number of input errors and $v_2$ is the number of faults.

\newpage

\end{document}